\documentclass[format=acmsmall, review=false]{acmart}
\usepackage{acm-ec-20}
\usepackage{booktabs} \usepackage[ruled, linesnumbered]{algorithm2e} 
\SetAlFnt{\small}
\SetAlCapFnt{\small}
\SetAlCapNameFnt{\small}
\SetAlCapHSkip{0pt}
\IncMargin{-\parindent}
\graphicspath{{figs/}}

\usepackage{microtype}
\usepackage{mathtools}
\usepackage{subcaption}
\usepackage{cleveref}
\usepackage{nicefrac}
\usepackage{hyperref}
\usepackage{xspace}
\usepackage{stackrel}
\usepackage{graphicx}
\usepackage{soul} 

\usepackage{tikz}
\usetikzlibrary{positioning}

\setcitestyle{authoryear}

\title{How to Fairly Allocate Easy and Difficult Chores}

\author{Soroush Ebadian}
\email{soroush@cs.toronto.edu}
\affiliation{\institution{University of Toronto}}
\author{Dominik Peters}
\email{dominik@cs.toronto.edu}
\affiliation{\institution{University of Toronto}}
\author{Nisarg Shah}
\email{nisarg@cs.toronto.edu}
\affiliation{\institution{University of Toronto}}

\DeclarePairedDelimiter\floor{\lfloor}{\rfloor}
\DeclarePairedDelimiter\len{\lvert}{\rvert}
\DeclarePairedDelimiter\abs{\lvert}{\rvert}
\DeclarePairedDelimiter\set{\{}{\}}

\newcommand{\Z}{\mathbb{Z}}
\newcommand{\N}{\mathbb{N}}
\newcommand{\R}{\mathbb{R}}
\renewcommand{\ge}{\geqslant}

\renewcommand{\le}{\leqslant}

\renewcommand{\hat}{\widehat}

\newcommand{\xx}{\mathbf{x}}
\newcommand{\yy}{\mathbf{y}}
\providecommand\vv{} \renewcommand{\vv}{\mathbf{v}}

\DeclareMathOperator{\MPB}{MPB}
\DeclareMathOperator{\PB}{PB}

\newcommand{\noviol}{\operatorname{NoViol}}

\DeclareMathOperator*{\argmin}{\arg\min}

\DeclareMathOperator{\poly}{\operatorname{poly}}
\DeclareMathOperator{\level}{\operatorname{level}}

\DeclareMathOperator{\entitled}{\operatorname{entitled}}

\numberwithin{claim}{section}

\numberwithin{observation}{section}

\newcommand{\ls}{\operatorname{ls}\xspace}

\newcommand{\agents}{\mathcal{N}}
\newcommand{\items}{\mathcal{M}}
\newcommand{\partitions}{\mathcal{P}}

\newcommand{\MMS}{\operatorname{MMS}}

\newcommand{\price}{\mathbf{p}}
\newcommand{\priceone}{\price_{\textup{up to 1}}}

\newcommand{\idle}{{\operatorname{idle}}}

\begin{abstract}
A major open question in fair allocation of indivisible items is whether there always exists an allocation of chores that is Pareto optimal (PO) and envy-free up to one item (EF1). We answer this question affirmatively for the natural class of bivalued utilities, where each agent partitions the chores into easy and difficult ones, and has cost $p > 1$ for chores that are difficult for her and cost $1$ for chores that are easy for her. Such an allocation can be found in polynomial time using an algorithm based on the Fisher market.

We also show that for a slightly broader class of utilities, where each agent $i$ can have a potentially different integer $p_i$, an allocation that is maximin share fair (MMS) always exists and can be computed in polynomial time, provided that each $p_i$ is an integer. Our MMS arguments also hold when allocating goods instead of chores, and extend to another natural class of utilities, namely weakly lexicographic utilities. \end{abstract}

\begin{document}

\maketitle

\section{Introduction}\label{sec:intro}

Fair allocation of collective resources and burdens
between agents is a fundamental task in multi-agent systems. Everyday applications include splitting an estate between heirs or joint assets between a divorcing couple (resources), or splitting work shifts between staff or household chores between roommates (burdens). 

We are interested in indivisible resources and burdens (i.e., ones that cannot be subdivided). Let $\items$ be the set of such \emph{items}. Following a canonical model, we assume that each agent $i$ has a valuation $v_i(r)$ for each item $r \in \items$. This gives rise to an \emph{additive} utility function over bundles of items: Agent $i$'s utility for a bundle $S \subseteq \items$ is $v_i(S) = \sum_{r \in S} v_i(r)$. Items are called \emph{goods} if all agents have non-negative valuations for them, and \emph{chores} if all agents have non-positive valuations for them. We will only study cases where either all items are goods, or all items are chores. 
The goal is to find an \emph{allocation} $\xx$, which is a partition of the set $\items$ of items between the agents, with $\xx_i$ denoting the bundle allocated to agent $i$. 
An allocation is \emph{efficient} or \emph{Pareto optimal} (PO) if there is no other allocation $\yy$ which every agent $i$ weakly prefers to $\xx$ (i.e., $v_i(\yy_i) \ge v_i(\xx_i)$), and for which at least one of these inequalities is strict. We are interested in finding allocations that are efficient and also \emph{fair}.
In particular, we will look at restricted classes of utilities that allow us to guarantee stronger fairness axioms than the state of the art for general additive utilities.

\subsection{Envy-Freeness Up To One Item (EF1)}
Perhaps the most compelling fairness guarantee from the literature is \emph{envy-freeness} (EF)~\cite{GS58,Fol67}, which demands that no agent envy another agent (i.e., $v_i(\xx_i) \ge v_i(\xx_j)$ for all agents $i,j$).  
However, it is easy to see that envy-freeness cannot be guaranteed; if we are allocating a single item between two agents, then one will necessarily envy the other. In response, the literature has turned to relaxations which require that agents not envy others by too much.
A particularly appealing axiom is called \emph{envy-freeness up to one item} (EF1)~\citep{Bud11}, which demands that envy between any two agents be avoidable by the removal of a single item from the bundle of one of the two agents. 
For allocating goods, \citet{CKMP+19} show that an elegant rule called \emph{maximum Nash welfare} (MNW) satisfies EF1 and PO simultaneously. Informally, this rule maximizes the product of utilities of the agents for their assigned bundles, i.e., $\prod_i v_i(\xx_i)$. 
Due to its attractive properties, this rule has been deployed to the popular fair division website \href{http://www.spliddit.org}{Spliddit.org}, where it has been used by more than 10,000 people for applications such as dividing estates and settling divorces~\cite{Shah17}.
Unfortunately, MNW has no natural equivalent for chores, and whether an EF1 and PO allocation of chores always exists has remained a major open question. 

To make progress in resolving this problem, we look towards restricted families of utility functions. An example is the class of \emph{binary} utilities, in which all valuations are in $\set{0,-1}$. For allocating goods, the corresponding class of $\set{0,1}$-utilities is interesting and well-understood~\cite{barman2018greedy,halpern2020fair}. But for allocating chores, this class is trivial: first allocate any chore for which some agent has utility $0$ to such an agent; then all agents have utility $-1$ for all remaining chores, and we can allocate them as equally as possible to obtain an EF1 + PO allocation. 

A larger class is that of \emph{bivalued} utilities, where all valuations are in $\set{a,b}$, for some fixed $0 > a > b$. The corresponding class for goods (with $0 < a < b$) has already received significant attention in the literature~\cite{aziz2020random,garg2021computing,akrami2021nash}, where it has been used to achieve fairness guarantees stronger than EF1~\cite{amanatidis2021maximum}. This class seems interesting for practical applications: when eliciting agent preferences, it is often cumbersome for agents to submit exact numerical utilities. Instead, it is much easier to ask each agent to classify chores into easy and difficult ones, with an interface familiar from approval voting. Then, one can fix reasonable values of $a$ and $b$, and assume that all agents have utility $a$ for the chores they consider to be easy and $b$ for the ones they consider to be difficult. 

For our results, scaling an agent's utilities multiplicatively makes no difference. Hence, bivalued utilities for chores can also be thought of as having utilities in $\set{-1,-p}$ for some number $p = \frac{b}{a} > 1$ (or $\set{1,p}$ for goods). Our main contribution is to show that EF1 and PO allocations of chores always exist under bivalued utilities, and that such an allocation can be found in polynomial time. 
We obtain this result via an algorithm based on Fisher markets. Our algorithm borrows some ideas from the existing Fisher-market-based algorithm for finding an EF1 and PO allocation of goods~\cite{barman2018finding,garg2021computing}, but combines it with  a more intricate analysis and new techniques that are key to making the algorithm work for chores. 
In simultaneous independent work, \citet{bivaluedAAAI} obtained the same result, also via Fisher markets.

\subsection{Maximin Share Fairness (MMS)}

In addition to envy-freeness up to one good (EF1), we consider another popular relaxation of envy-freeness called \emph{maximin share fairness} (MMS)~\cite{Bud11}. This notion wants to give each agent at least as much utility as the maximum that the agent can achieve by partitioning the items into $n$ bundles and receiving her least preferred bundle from that partition. For general additive valuations, an MMS allocation may not exist, both for goods \citep{kurokawa2018fair} and for chores \citep{aziz2017algorithms}. Thus, we again turn to restricted utility classes that allow us to guarantee MMS.

We first consider \emph{personalized bivalued utilities}, where the valuations of each agent $i$ for the chores (resp., goods) lie in $\set{-1,-p_i}$ (resp., $\set{1,p_i}$) for some $p_i > 1$. In contrast to bivalued utilities, the value $p_i$ can differ between agents. To elicit such valuations, one can ask each agent $i$ to first partition the chores into easy and difficult ones (resp., goods into ordinary and preferred ones) using an approval interface, and then submit a number $p_i$ indicating how many easy chores they would do instead of a single difficult one (resp., how many ordinary goods they would be willing to take in place of a single preferred one).
We show that for personalized bivalued utilities, for both goods and chores, an allocation satisfying MMS always exists and can be computed in polynomial time, provided that $p_i$ is an \emph{integer} for each agent $i$. 
Integrality would be the natural outcome of the aforementioned elicitation.
Whether MMS can be guaranteed for non-integral $p_i$ remains an open question.
For (non-personalized) bivalued utilities (with integer $p$), we show that we can compute in polynomial time an MMS allocation that is also PO.

We also prove the existence of MMS allocations for another class of utilities, namely \emph{weakly lexicographic} utilities, for both goods and chores. 
Weakly lexicographic utilities are a natural assumption if valuations are elicited by a system that asks agents to rank the items in order of desirability, allowing for ties.
The defining assumption is that an agent likes each good (resp., dislikes each chore) more than all strictly less preferred goods (resp., more preferred chores) combined. 
Such utility functions, which we refer to as weakly lexicographic utilities, have been considered in the literature~\cite{aziz2019efficient}. We prove that for these utilities, an allocation that satisfies both MMS and PO always exists and can be computed in polynomial time. \citet{HSVX21} prove this for the special case of allocating goods under (strictly) \emph{lexicographic} utilities (in which there are no ties); our result extends theirs to allocating goods or chores under weakly lexicographic utilities. 

Both of our MMS existence results depend on a simple algorithm for computing MMS values (i.e., the utility value guaranteed by the MMS property on a given instance). Computing these values is NP-hard for both goods and chores under general additive valuations (being a special case of the \textsc{3-Partition} problem~\citep[p.~224]{GJ79}), but we show that it can be done in polynomial time for \emph{factored} utility functions, which includes both personalized bivalued and weakly lexicographic utilities as special cases. A utility function is factored if the non-zero utility values, say $p_1,\ldots,p_k$, that it uses are such that $p_{j+1}$ is an integer multiple of $p_j$ for each $j \in [k-1]$; 
examples are $\{1,2,6,12\}$-valuations and $\{0,-1, -5, -45\}$-valuations.

\Cref{fig:hasse} shows the utility classes that we study, together with inclusion relationships and relevant results, both known and new.

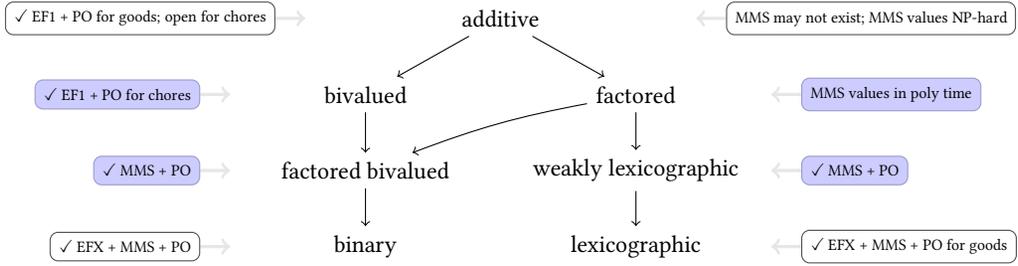
\begin{figure}
\centering
\begin{tikzpicture}
	[
	valuation class/.style={font=\small},
	known result/.style={font=\tiny, draw=black!70, rounded corners=3pt},
	new result/.style={font=\tiny, draw=blue!50!black!40!white, fill=blue!20!white, rounded corners=3pt},
	right col/.style={anchor=west},
	left col/.style={anchor=east},
	result pointer/.style={->, very thick, draw=black!10}]
	\node[valuation class] (add) at (3,0) {additive};
	\node[valuation class] (biv) at (1.2,-1) {bivalued};
	\node[valuation class] (fact) at (4.8,-1) {factored};
	\node[valuation class] (fact-biv) at (1.2,-2) {factored bivalued};
	\node[valuation class] (wolex) at (4.8,-2) {weakly lexicographic};
	\node[valuation class] (binary) at (1.2,-3) {binary};
	\node[valuation class] (lex) at (4.8,-3) {lexicographic};
	
	\path[->] (add)      edge (biv)
					     edge (fact)
			  (biv)      edge (fact-biv)
			  (fact-biv) edge (binary)
			  (fact)     edge (wolex)
			             edge[bend right=5] (fact-biv)
			  (wolex)    edge (lex);
	
	\draw[result pointer] (6,0) -- (5.6,0);
	\draw[result pointer] (0,0) -- (0.4,0);	
	\foreach \y in {-1,-2,-3} {
		\draw[result pointer] (7,\y) -- (6.6,\y);
		\draw[result pointer] (-1,\y) -- (-0.6,\y);
	}
	
	\node[known result, left col] at (0,0) (fact-mms-comp) {$\smash\checkmark$ EF1 + PO for goods; open for chores};
	\node[new result, left col] at (-1,-1) (fact-mms-comp) {$\smash\checkmark$ EF1 + PO for chores};
	\node[new result, left col] at (-1,-2) (fact-mms-comp) {$\smash\checkmark$ MMS + PO};
	\node[known result, left col] at (-1,-3) (fact-mms-comp) {$\smash\checkmark$ EFX + MMS + PO};
	
	\node[known result, right col] at (6,0) (fact-mms-comp) {MMS may not exist; MMS values NP-hard};		  
	\node[new result, right col] at (7,-1) (fact-mms-comp) {MMS values in poly time};
	\node[new result, right col] at (7,-2) (fact-mms-comp) {$\smash\checkmark$ MMS + PO};
	\node[known result, right col] at (7,-3) (fact-mms-comp) {$\smash\checkmark$ EFX + MMS + PO for goods};
\end{tikzpicture}
\caption{Hasse diagram of valuation classes and results. Shaded blue nodes are new results of this paper, boxed results are known. Checkmarks ($\smash\checkmark$) denote existence results, which all come with polynomial-time algorithms. Results hold for both goods and chores unless otherwise indicated.}
\label{fig:hasse}
\end{figure}

\subsection{Related Work}\label{sec:related-work}

Let us summarize a few related threads of work on fair allocation of goods and chores to better contextualize our contributions. 

\medskip\noindent\textbf{Fisher market.} As mentioned earlier, we achieve our main result --- an EF1 + PO allocation of chores with bivalued utilities --- using the framework of Fisher markets and competitive equilibria. Fisher markets are typically studied for items that are \emph{divisible}, i.e., that can be portioned out fractionally between the agents. In this case, a Fisher market equilibrium allocation exists and is EF + PO~\cite{Var74}. For goods, these allocation happen to be those that maximize the Nash welfare, and they can be computed in strongly polynomial time~\cite{devanur2008market,orlin2010improved}. For chores, the set of equilibria has a more intricate structure~\cite{bogomolnaia2017competitive} and their computation is an open question~\cite{branzei2019algorithms}; \citet{boodaghians2022polynomial} design an FPTAS for this problem. One issue with chore allocation is that neither minimizing nor maximizing the product of agents' costs for their assigned bundles (the equivalent of the Nash welfare objective for goods) yields a desirable allocation. However, \citet{bogomolnaia2017competitive} show that maximizing this objective \emph{subject to PO} yields one of the aforementioned equilibria. Unfortunately, the natural analog of this rule for indivisible chores fails EF1, even for bivalued utilities.\newcommand{\mnw}[1]{\underline{#1}}\footnote{An example with $4$ agents and $8$ items has valuations 
$(-4, -4, -1, -1, \mnw{-1}, \mnw{-1}, \mnw{-1}, \mnw{-1})$,
$(-4, -4, \mnw{-1}, \mnw{-1}, -4, -4, -4, -4)$,
$(-4, \mnw{-4}, -4, -4, -1, -1, -4, -4)$, and
$(\mnw{-4}, -4, -4, -4, -4, -4, -1, -1)$, 
where underlined entries indicate items allocated by the rule described in the text. Under this allocation, the first agent envies the second, even up to one item.}
\citet{barman2018finding} adapt Fisher markets to indivisible goods. They use this framework to show that an EF1 + PO allocation can be found in pseudo-polynomial time. \citet{garg2021computing} improve the running time to strongly polynomial when each agent has at most polynomially many utility levels across all bundles of goods. 
The Fisher market approach has also been used to obtain efficient allocations that are proportional up to one item (PROP1) for both goods \citep{barman2019proximity} and chores \citep{branzei2019algorithms}.

\medskip\noindent\textbf{Factored bivalued utilities and max Nash welfare.} The special case of bivalued utilities in which the utility values lie in $\set{a,b}$ for $|a| < |b|$ and $b/a$ is an \emph{integer} (which we refer to as factored bivalued utilities) has been studied in the context of allocating goods. The maximum Nash welfare (MNW) rule is NP-hard to compute for general additive utilities~\cite{CKMP+19}, while \citet{barman2018greedy} show that it can be computed in polynomial time for binary ($\set{0,1}$) utilities. For bivalued utilities, its computability was an open question until recently when \citet{akrami2021nash} established a surprising dichotomy: it is polynomial-time computable when $b/a$ is an integer (factored bivalued utilities) but NP-hard to compute when $a$ and $b$ are coprime.

\medskip\noindent\textbf{Existential results in restricted cases.} Our approach of resolving the question of EF1+PO allocation of chores --- open for general additive utilties --- under restricted settings is reminiscent of recent advances that achieve similar goals for other open questions. For example, for allocating goods, envy-freeness up to any good (EFX)~\cite{CKMP+19} is a fairness property stronger than EF1, which demands that it be possible to remove envy between any two agents by removing \emph{any} good from the bundle of the envied agent. It is an open question whether an EFX allocation of goods always exists for general additive utilities, and recent advances has resolved this positively under restricted cases of bivalued utilities~\cite{amanatidis2021maximum}, identical utilities~\cite{plaut2020almost}, and three agents with general additive utilities~\cite{chaudhury2020efx}.

\medskip\noindent\textbf{MMS.} For allocating goods, \citet{kurokawa2018fair} show that there exists an instance with additive utilities in which no allocation satisfies MMS. This motivates two threads of work. One, similarly to our work, focuses on establishing the existence (and sometimes efficient computability) of MMS allocations under restricted utility classes such as utility functions with identical multisets~\cite{BL16}, (strictly) lexicographic utilities~\cite{HSVX21}, and ternary ($\set{0,1,2}$) utilities~\cite{amanatidis2021maximum}. Also, note that factored bivalued utilities include $\set{1,2}$-utilities as a special case, and, since we argued in the introduction that $0$ utilities can be easily addressed for chores, our MMS result in this case mirrors that of \citet{amanatidis2021maximum}. The other thread focuses on approximating the MMS guarantee for general additive utilities: the best known multiplicative approximations are (slightly better than) $3/4$ for goods~\cite{garg2021improved} and $9/11$ for chores~\cite{huang2021algorithmic}.
 \section{Preliminaries}\label{sec:prelim}

For $k \in \N$, define $[k] = \set{1,\ldots,k}$. 

\medskip\noindent\textbf{Instances:} A \emph{fair division instance} is given by $I = (\agents,\items,\vv)$, where $\agents = [n]$ is a set of $n$ agents, $\items$ is a set of $m$ indivisible items, and $\vv = (v_1,\ldots,v_n)$ is the utility profile with $v_i : \items \to \R$ being the utility function of agent $i$ and $v_i(r)$ indicating $i$'s utility for item $r$. 

In this work, we assume that either all items are \emph{goods} for all agents (i.e., $v_i(r) \ge 0$ for all $i \in \agents$ and $r \in \items$), in which case we refer to $I$ as a \emph{goods division} instance, or all items are \emph{chores} for all agents (i.e., $v_i(r) \le 0$ for all $i \in \agents$ and $r \in \items$), in which case we refer to $I$ as a \emph{chore division} instance.

We focus our attention to the class of additive utility functions, in which the utility of agent $i$ for a set of items $S \subseteq \items$ is given by, with slight abuse of notation, $v_i(S) = \sum_{r \in S} v_i(r)$. We are interested in the following subclasses of additive utilities. Let $v$ denote an additive utility function over a set of items $\items$ in a goods division or chore division instance. 

\begin{definition}[Factored utilities]\label{def:factored}
	We say that a utility function $v: \items \to \set{0,p_1,\ldots,p_k} \subset \Z$ is \emph{factored} if $p_j$ divides $p_{j+1}$ (i.e., $p_{j+1} = q \cdot p_j$ for some $q \in \N_{> 0}$) for each $j \in [k-1]$. 
\end{definition}

\begin{definition}[Weakly lexicographic utilities]\label{def:wolex}
	We say that $v$ is \emph{weakly lexicographic} if there is a partition $(L_1,\ldots,L_k)$ of $\items$ with
	\begin{enumerate}
		\item $\forall i \in [k]$ and $r,r' \in L_i$, we have $|v(r)| = |v(r')| > 0$, and
		\item $\forall i \in [k]$ and $r \in L_i$, we have $\abs{v(r)} > \abs{\sum_{r'\in L_{i + 1} \cup \ldots \cup L_{k}}v(r')}$.
	\end{enumerate}
	Further, if $k=m$, then we say that $v$ is (strictly) lexicographic. 
\end{definition}

Weakly lexicographic utilities can be seen as a special case of factored utilities, as we may assume that $|v_i(r)|$ is a power of $m$. The following lemma shows that we can make that assumption without changing the ordinal preferences over bundles. 

\begin{lemma}\label{lem:wolex-factored}
	Let $v$ be a weakly lexicographic utility function over a set of items $\items$. Then, there exists a weakly lexicographic factored utility function $v'$ given by $v' : \items \to \set{1,m,m^2,\ldots}$ for goods or $v' : \items \to \set{-1,-m,-m^2,\ldots}$ for chores such that $v(S) \le v(S') \Leftrightarrow v'(S) \le v'(S')$ for all $S,S' \subseteq \items$. 
\end{lemma}
\begin{proof}
	Let $(L_1,\ldots,L_k)$ be the partition of $\items$ under $v$ as in \Cref{def:wolex}. Let $S, S' \subseteq \items$ be two arbitrary subsets of items that $v(S) \le v(S')$. Suppose $v$ is a valuation function for goods.
	
	If $v(S) = v(S')$, then for all $i \in [k],$ $\len{S \cap L_i} = \len{S' \cap L_i}$.
	Therefore, $v'(S) = \sum_{i \in [k]} \len{S \cap L_i} \cdot m^i =
	\sum_{i \in [k]} \len{S' \cap L_i} \cdot m^i = v'(S')$. 
	
	If $v(S) < v(S')$, then there exists an $i \in [k]$, such that $\len{S \cap L_{i}} < \len{S' \cap L_{i}}$, and for all $ i' > i$, $\len{S \cap L_{i'}} = \len{S' \cap L_{i'}}$. Then,
	\begin{equation*}
		\scriptsize
		v'(S') - v'(S) = \sum_{j \in [i]} \left(\len{S' \cap L_j} - \len{S \cap L_j}\right) \cdot m^j \ge m^i - \sum_{j \in [i - 1]}  \len{S \cap L_j} \cdot m^j  \ge m^i - (m - 1) \cdot m^{i - 1} > 0.
	\end{equation*}
	The proof for the chores case is similar.
\end{proof}
 
\begin{definition}[Bivalued utilities]
	We say that $v$ is \emph{bivalued} if there are non-zero $a,b \in\R$ such that $v(r) \in \set{a,b}$ for all $r \in \items$. In case of goods, we will use the convention $0 < a < b$, and in case of chores, we will use the convention $0 > a > b$. Further, if $a$ divides $b$, we say that $v$ is \emph{factored bivalued}. 
\end{definition}

We say that a goods division or chore division instance has factored (resp., weakly lexicographic) utilities if every agent has a factored (resp., weakly lexicographic) utility function. We say that the instance has bivalued utilities if all agents have bivalued utilities for some common $a,b$ (i.e., there exist $a,b$ such that $v_i(r) \in \set{a,b}$ for all $i,r$). We say that the instance has \emph{personalized bivalued} utilities if each agent $i$ has a bivalued utility function (perhaps with personalized $a_i,b_i$).\footnote{Personalized bivalued utilities are a special case of what \citet{garg2021fair} call $k$-ary utilities.}

\medskip\noindent\textbf{Allocations:} An \emph{allocation} $\xx = (\xx_1, \dots, \xx_n)$ is a collection of bundles $\xx_i \subseteq \items$, one for each agent $i\in\agents$, such that the bundles are pairwise disjoint ($\xx_i \cap \xx_j = \emptyset$ for all distinct $i,j \in \agents$) and every item is allocated ($\bigcup_{i\in\agents} \xx_i = \items$).

\medskip\noindent\textbf{Fairness and Efficiency Desiderata:} 
We study two prominent fairness notions for the allocation of indivisible items, known as envy-freeness up to one item~\cite{LMMS04,Bud11,CKMP+19} and maximin share fairness~\cite{Bud11,kurokawa2018fair}. These are respectively relaxations of the classical notions of envy-freeness and of proportionality. We give definitions that work for both goods and chores \cite{aziz2022fair}.

\begin{definition}[Envy-freeness up to one item]
	An integral allocation $\xx$ is said to be \emph{envy-free up to one item} (EF1) if, for every pair of agents $i, j \in \agents$ such that $\xx_i \cup \xx_{j} \ne \emptyset$, there exists an item $r \in \xx_i \cup \xx_{j}$ such that $v_i(\xx_i \setminus \set{r}) \ge v_i(\xx_{j} \setminus \set{r})$. 
\end{definition}

In a goods division problem, this reduces to $v_i(\xx_i) \ge v_i(\xx_{j} \setminus \set{g})$ for some good $g \in \xx_{j}$ (a good removed from the bundle of agent $j$), while in a chore division problem, it reduces to $v_i(\xx_i \setminus \set{c}) \ge v_i(\xx_{j})$ for some $c \in \xx_i$ (a chore removed from the bundle of agent $i$).

\begin{definition}[Maximin share fairness]
	For $k \in \N$, let $\partitions^k(\items)$ be the set of all partitions of $\items$ into $k$ bundles. For agent $i \in \agents$, let
	\[
	\textstyle \MMS^k_i = \max_{(S_1,\ldots,S_k) \in \partitions^k(\items)} \min_{t \in [k]} v_i(S_t).
	\]
	Note that this is the maximum utility she can obtain by partitioning the items into $k$ bundles and receiving the least valued bundle. We refer to an optimal partition $(S_1,\ldots,S_k)$ in the above equation as a maximin $k$-partition for agent $i$. The \emph{maximin share} of agent $i \in \agents$ is defined as $\MMS^n_i$. For simplicity of notation, we write $\MMS^n_i$ as $\MMS_i$ and refer to a maximin $n$-partition as a \emph{maximin partition}. An allocation $\xx$ is said to be maximin share fair (MMS) if each agent receives at least as much utility as her maximin share, i.e., if $v_i(\xx_i) \ge \MMS_i$ for each agent $i \in \agents$. 
\end{definition}

Finally, we define a prominent notion of economic efficiency. 
\begin{definition}[Pareto optimality]
	We say that allocation $\xx$ is \emph{Pareto dominated} by allocation $\xx'$ if $v_i(\xx_i) \le v_i(\xx'_i)$ for every agent $i \in \agents$ and at least one inequality is strict. An allocation $\xx$ is said to be \emph{Pareto optimal} (PO) if it is not Pareto dominated by any allocation.
\end{definition}

 \section{EF1 + PO for Bivalued Chores}
\label{sec:ef1}

In this section, we present a polynomial-time algorithm that finds an EF1 and PO allocation for chore division instances with bivalued utilities, thereby also establishing the existence of such allocations. Specifically, we scale agent utilities such that for some $p > 1$, the utility of each agent $i$ for every chore $c$ is $v_i(c) \in \set{-1,-p}$. Further, if some agent $i$ has $v_i(c) = -p$ for all chores $c$, then we will scale this so that $v_i(c) = -1$ for all chores $c$. This will ensure that each agent values at least one chore at $-1$. Recall that scaling the utilities of any agent does not affect whether an allocation is EF1 or PO. 

Our algorithm builds on the algorithm by \citet{barman2018finding} for finding an EF1 and PO allocation of goods. Their algorithm starts with a PO allocation and then moves items around until it is EF1, while maintaining that the allocation is PO at every step. Pareto optimality is maintained in the algorithm by ensuring that the allocation remains an equilibrium in a \emph{Fisher market}. Thus, we start by introducing some basic concepts about Fisher markets.

\subsection{Fisher Markets for Chore Division}

A \emph{price vector} $\price$ assigns a price $\price(c) > 0$ to each chore $c$. For a subset $S \subseteq \items$ of chores, we write $\price(S) = \sum_{c \in S} \price(c)$. Given this price vector, the \emph{pain per buck} (PB) ratio of agent $i$ for chore $c$ is defined as $\PB_i(c) = \frac{\abs{v_i(c)}}{ \price(c)}$, and the \emph{minimum pain per buck} (MPB) ratio of agent $i$ is defined as $\MPB_i = \min_{c \in \items} \PB_i(c)$. A chore $c$ with $\PB_i(c) = \MPB_i$ is called an \emph{MPB chore} for agent $i$.
\begin{definition}
	A pair $(\xx,\price)$ of an allocation $\xx$ and a price vector $\price$ is a (Fisher market) \emph{equilibrium}\footnote{This is sometimes called a \emph{quasi-equilibrium}, because we do not specify an exogenous budget for each agent.}
if each agent is allocated only her MPB chores, i.e., if $\PB_i(c) = \MPB_i$ for all $i \in \agents$ and all $c \in \xx_i$.
\end{definition}

We say that $\xx$ is an equilibrium allocation if $(\xx,\price)$ is an equilibrium for some price vector $\price$. The following is known to hold by the so-called first welfare theorem.

\begin{proposition}
	Every equilibrium allocation is Pareto optimal.
\end{proposition}
\begin{proof}
	Let $(\xx, \price)$ be an equilibrium. Suppose $c \in \xx_i$. Then $\MPB_i = \PB_i(c)$ and hence, remembering that $v_i(c)$ is negative, we have $v_i(c) / \MPB_i = v_i(c) / \PB_i(c) = -\price(c)$. On the other hand, if $c \not \in \xx_i$, then $\MPB_i \le \PB_i(c)$ and hence $v_i(c) / \MPB_i \le v_i(c) / \PB_i(c) = -\price(c)$.
	
	From this it follows that if $c \in \xx_i$, then we have $v_i(c)/\MPB_i \ge v_j(c)/\MPB_j$ for all $j \in \agents$.
	Hence $\xx$ maximizes the value $\sum_{i \in \agents} v_i(\xx_i)/\MPB_i$. But any Pareto improvement over $\xx$ would strictly increase this value (noting that $\MPB_i > 0$), so $\xx$ must be Pareto optimal.
\end{proof}

In fact, a stronger statement is true: every equilibrium allocation is \emph{fractionally Pareto optimal} (fPO), which means it is not even Pareto dominated by a \emph{fractional} allocation \citep{barman2018finding}. Moreover, a second welfare theorem shows that every fPO allocation arises as an equilibrium allocation \citep{barman2018finding}. For the special case of bivalued utilities, we prove in \Cref{sec:fpo} that PO and fPO are equivalent. 
\begin{theorem}
	\label{thm:fpo}
	Given a goods or chore division problem with bivalued utilities, an allocation $\xx$ is Pareto optimal	 if and only if it is fractionally Pareto optimal.
\end{theorem}
Thus, the stronger efficiency guarantee fPO that Fisher markets provide does not have bite in our setting. Conversely, though, this equivalence means that any method that identifies a PO allocation for bivalued utilities must implicitly calculate an equilibrium, so in a sense the Fisher market method is the most natural approach for constructing PO allocation for bivalued utilities.

As an invariant, our algorithm will keep the considered allocation an equilibrium. Our aim is to find a fair equilibrium, by which we will mean that the prices of agents' bundles are approximately equal. This notion is an adaption to the chores case of a property introduced by \citet{barman2018finding}. 

\begin{definition}[Price envy-freeness up to one item] We say that $(\xx,\price)$ is \emph{price envy-free up to one item} (pEF1) if, for all $i,j \in \agents$ with $\xx_i \ne \emptyset$, there is a chore $c \in \xx_i$ such that $\price(\xx_i \setminus \set{c}) \le \price(\xx_{j})$.
\end{definition}

Like for the goods division case~\cite{barman2018finding}, pEF1 implies EF1.
\begin{lemma}
	\label{lem:pef1-gives-ef1}
	If $(\xx,\price)$ is a pEF1 equilibrium, then $\xx$ is EF1.
\end{lemma}
\begin{proof}
	Fix a pair of agents $i,j \in \agents$. We want to show that $v_i(\xx_i) \ge v_i(\xx_j)$. If $\xx_i = \emptyset$, this holds trivially. Otherwise, pEF1 indicates that there exists a chore $c \in \xx_i$ such that $\price(\xx_i \setminus \set{c}) \le \price(\xx_j)$. Then, using the definition of $\PB_i$ and $\MPB_i$, we have
	\begin{align*}
		|v_i(\xx_i \setminus \set{c})| = \MPB_i \cdot \price(\xx_i \setminus \set{c}) \le \MPB_i \cdot \price(\xx_j) \le \sum_{c' \in \xx_j} \PB_i(c') \cdot \price(c') = |v_i(\xx_j)|,
	\end{align*}
	where the first transition uses the fact that in an equilibrium allocation agent $i$ is only assigned her MPB chores. Hence, we have $v_i(\xx_i \setminus \set{c} \ge v_i(\xx_j)$, as needed.
\end{proof}
 
For $S \subseteq \items$, define $\priceone(S) = \price(S) - \max_{c \in S} \price(c)$, if $S \neq \emptyset$, and $0$ if $S = \emptyset$. We often write $\ls\in\agents$ for the \emph{least spender}, i.e., an agent $\ls \in \argmin_{i \in \agents} \price(\xx_i)$. Then we see that $(\xx, \price)$ is pEF1 if and only if $\priceone(\xx_i) \le \price(\xx_{\ls})$ for all $i \in \agents$. Let us call an agent $i \in \agents$ a \emph{violator} if $\priceone(\xx_i) > \price(\xx_{\ls})$. Thus, $(\xx, \price)$ is pEF1 if and only if no agent is a violator.

Given an equilibrium $(\xx, \price)$, we write $j \stackrel{c}{\gets} i$ if agent $i$ owns item $c$ (so $c\in\xx_i$) and $c$ is an MPB chore for $j$. Thus, if we have $j \stackrel{c}{\gets} i$ then the allocation $\xx'$ obtained from $\xx$ by transferring item $c$ from $i$ to $j$ is still an equilibrium.

\newcommand{\reachable}{\mathrel{\reflectbox{$\leadsto$}}}
\begin{definition}[MPB alternating path]
	An \emph{MPB alternating path} of length $\ell$ from $i_\ell$ to $i_0$ is a sequence $i_0 \stackrel{c_1}{\gets} i_1 \stackrel{c_2}{\gets} \cdots \stackrel{c_\ell}{\gets} i_\ell$.
\end{definition}

If there exists an MPB alternating path from $i_\ell$ to $i_0$, we write $i_0 \reachable i_\ell$. We always have $i_0 \reachable i_0$.

\subsection{Algorithm}

We now present \Cref{alg:full-algo} which computes an PO and EF1 allocation given a chore division instance with bivalued utilities.

\begin{algorithm}
	\DontPrintSemicolon
	\SetAlgoNoEnd
	\SetAlgoLined
	\SetKwProg{PhaseI}{Phase 1}{}{}
	\SetKwProg{PhaseII}{Phase 2}{}{}
	\SetKwProg{PhaseIIa}{Phase 2a}{}{}
	\SetKwProg{PhaseIIb}{Phase 2b}{}{}
	\SetKwProg{PhaseIII}{Phase 3}{}{}
	\caption{EF1 + PO for Bivalued Chores}
	\label{alg:full-algo}
	\PhaseI{Initialization}{
		Let $\xx$ be an allocation maximizing social welfare $\sum_{i\in\agents} v_i(\xx_i)$.\;
		For each $c \in \items$, let $\price_c = p \cdot \abs{ \max_{i \in \agents} v_i(c) }$\;
		\label{line:initialize-prices}
		$k \gets 1$, the number of the current iteration\;
	}
	\PhaseIIa{Reallocate chores}{
		\label{line:phase-2-start}
		\For{$\ell \in (k - 2, k - 3, \ldots, 2, 1)$}{
			\While{true}{
				\label{line:phase-2a-start}
				$i \gets$ an agent from $  \arg\max_{i\in H_\ell} \priceone(\xx_i)$\;
				$j \gets$ an agent from $  \arg\min_{j \in H_{\ell + 1} \cup \dots \cup H_{k-1}} \price(\xx_j)$\;\label{line:phase-2a-def-j}
				\eIf{$\priceone(\xx_i) > \price(\xx_j)$}{
					\label{line:phase-2a-violator-condition}
					$c \leftarrow $ any item from $\xx_i \setminus \entitled(i)$\; \label{line:non-entitled-item-exists}
					\label{line:phase-2a-transfer}
					Transfer $c$ from $i$ to $j$\;
				}{
					\textbf{break}\; \label{line:phase-2a-end}
				}
			}
		}
	}
	\PhaseIIb{Reallocate chores}{
		\While{true}{
			\label{line:phase-2b-start}
			$\ls \leftarrow$ an agent from $ \arg\min_{i\in\agents}\price(\xx_i)$ \;
			\eIf{there is an MPB alternating path $\smash{\ls \stackrel{c_1}{\gets} i_1 \stackrel{c_2}{\gets} \cdots \stackrel{c_\ell}{\gets} i_\ell}$ with $\priceone(\xx_{i_\ell}) > \price(\xx_{\ls})$}{
				\label{line:phase-2b-condition}
				Choose such a path of minimum length $\ell$\;
				\label{line:phase-2b-transfer}
				Transfer $c_\ell$ from $i_\ell$ to $i_{\ell-1}$\;
			}{
				\textbf{break}\; \label{line:phase-2b-end}
			}
		}
		\If{$\xx$ satisfies \textup{pEF1}}{
			\Return $\xx$\; \label{line:return}
		}
	}
	\PhaseIII{Price reduction}{
		\label{line:phase-3-start}
$H_k \gets \set{i \in \agents : \text{there is an agent $\ls \in \arg\min_{i\in\agents}\price(\xx_i)$ with} \ls \reachable i}$\;
		$\blacktriangleright$ Timestamp: $t_{k,b}$
		\label{line:timestamp_t_k-b}\;
		$\alpha \leftarrow \min \set{\PB_i(c)/\MPB_i : i \in H_k, c \in \bigcup_{j \in \agents\setminus H_k} \xx_j}$\;
			\label{line:setting-alpha}
		\For{$i \in H_k$} {
			$\entitled(i) \gets \xx_i$ \;
			\For{$c\in \xx_i$} {
				$\price_c \leftarrow \frac{1}{\alpha} \cdot \price_c$\;
				\label{line:update-prices}
			}
		}
		$\blacktriangleright$ Timestamp: $t_{k,a}$
		\label{line:timestamp_t_k-a}\;
		$k \gets k + 1$\;
		Start Phase 2a (i.e. go to line \ref{line:phase-2-start})\;
	}
\end{algorithm}

\begin{theorem}
	Given a chore division problem $I = (\agents, \items, \vv)$ with bivalued utilities, \Cref{alg:full-algo} finds a PO and EF1 allocation in $\poly(n, m)$ time.
	\label{theorem:bival-poly}
\end{theorem}

\newcommand{\hypdisjoint}{(H1)\xspace}
\newcommand{\hyptransferable}{(H2)\xspace}
\newcommand{\hypnotviolator}{(H3)\xspace}
\newcommand{\hypentitled}{(H4)\xspace}
\newcommand{\hypalpha}{(H5)\xspace}
\newcommand{\hypprices}{(H6)\xspace}
\newcommand{\hypmpb}{(H7)\xspace}

The algorithm starts with an $(\xx,\price)$ that is guaranteed to be an equilibrium.  Then, it proceeds in \emph{iterations}. The value $k$, maintained by the algorithm, signifies the current iteration number. In each iteration $k$, the algorithm goes through Phases 2a, 2b, and 3 (except that in the final iteration the algorithm terminates after Phase 2b). During Phases 2a and 2b, the algorithm keeps the price vector $\price$ fixed and updates the allocation $\xx$, and in the subsequent Phase 3, it then keeps the allocation $\xx$ fixed, identifies a certain set $H_k$ of agents and updates the price vector $\price$ by reducing the prices of the chores allocated to $H_k$ by a multiplicative factor $\alpha$. 

A key property of our algorithm is that it ensures that the sets $H_k$ are disjoint across different iterations. This helps prove that our algorithm always terminates after at most $n$ iterations, since each $H_k$ contains at least one agent. This property differentiates our algorithm from the algorithm of \citet{barman2018finding} for allocating goods and requires us to introduce Phase 2a, which is not present in their algorithm. Phase 2b, on the other hand, is very similar to Phase 2 in their algorithm. 

Another key ingredient of our algorithm is that once an agent $i$ is assigned to a set $H_k$, the chores assigned to $i$ at that time become \emph{entitled chores} of agent $i$, denoted $\entitled(i)$. These are the chores which went through a price reduction while they were allocated to agent $i$. Subsequently the algorithm will never move the entitled chores away from $i$. Finally, in order to reason about the equilibria at different times during the execution of the algorithm, we timestamp important steps of the algorithm: $t_{k,b}$ and $t_{k,a}$ denote the time right \emph{before} and right \emph{after} the execution of Phase 3 in iteration $k$. 

We prove the correctness of the algorithm by induction on $k$. Specifically, we prove that for all $k \ge 1$ such that \Cref{alg:full-algo} reaches time $t_{k,a}$,
\begin{enumerate}
	\item[\hypdisjoint] $H_k \cap H_\ell = \emptyset$ for all $1 \le \ell < k$.
	\item[\hyptransferable] During iteration $k$, each time the algorithm reaches \cref{line:non-entitled-item-exists}, there exists a chore $c \in \xx_i \setminus \entitled(i)$. All such chores are MPB chores for agent $j$.
	\item[\hypnotviolator] At time $t_{k,b}$, each $i \in H_1 \cup \cdots \cup H_k$ is not a violator, so $\priceone(\xx_i) \le \price(\xx_{\ls})$ where $\ls$ is the least spender.
	\item[\hypentitled] At time $t_{k,a}$, each $i \in H_1 \cup \cdots \cup H_{k}$ owns every entitled item, $\entitled(i) \subseteq \xx_i$.
	\item[\hypalpha] When \cref{line:setting-alpha} is reached during iteration $k$, $\alpha$ is set to $p$.
	\item[\hypprices] At time $t_{k,a}$, we have $\price(c) \in \{1, p\}$ for all $c\in\items$. If $\price(c) = 1$, then $c \in \entitled(i)$ for some $i \in H_1 \cup \cdots \cup H_k$.
	\item[\hypmpb] At time $t_{k,a}$, we have $\MPB_i = 1$ for all $i \in H_1 \cup \dots \cup H_k$, and $\MPB_i = 1/p$ for all other agents.
\end{enumerate}

Let us first check that these statements together imply that $(\xx, \price)$ remains an equilibrium throughout the execution of the algorithm, and that the algorithm terminates in polynomial time, in \cref{line:return}. Then $(\xx, \price)$ is an equilibrium satisfying pEF1, and thus we have found an PO and EF1 allocation, as required.

\begin{lemma}
	\label{lem:always-equilibrium}
	Assume that \hypdisjoint to \hypmpb hold for all $k \ge 1$ such that \Cref{alg:full-algo} reaches time $t_{k,a}$. Then, throughout the algorithm's execution, $(\xx, \price)$ is an equilibrium.
\end{lemma}
\begin{proof}
	After initialization in \cref{line:initialize-prices}, $(\xx, \price)$ is an equilibrium because if $c \in \xx_i$ then $v_i(c) = \max_{j\in \agents}v_j(c)$ since $\xx$ is welfare-maximizing. Hence $\PB_i(c) = |v_i(c)|/(p\cdot|v_i(c)|) = 1/p$. On the other hand if $d \not\in \xx_i$ then $v_i(d) \le \max_{j\in \agents}v_j(d)$ so $\PB_i(d) \ge 1/p$ by the same calculation. Hence $\MPB_i = 1/p$ and $c$ is an MPB chore for $i$. So, $(\xx, \price)$ is an equilibrium.
	
	Item transfers in \cref{line:phase-2a-transfer} of Phase 2a keep $(\xx, \price)$ in equilibrium because $c$ is an MPB chore for $j$ by \hyptransferable. Item transfers in \cref{line:phase-2b-transfer} of Phase 2b preserve equilibrium because $c_\ell$ is an MPB chore for $i_{\ell -1}$ by the definition of MPB alternating path.
	
	Finally, price changes in \cref{line:update-prices} of Phase 3 preserve equilibrium by the definition of $\alpha$. To see this, note that $\alpha \ge 1$ (because, as we have seen, when we set $\alpha$ we are currently in equilibrium, so always $\PB_i(c) \ge \MPB_i(c)$ and so $\PB_i(c) / \MPB_i(c) \ge 1$). Thus the price change reduces prices, and thus increases some pain-per-buck ratios. It follows that for all $i \in \agents \setminus H_k$, items owned by $i$ remain MPB items for $i$ (since $\MPB_i$ can only go up and the prices of chores owned by $i$ do not change). Now write $\MPB_i'$ and $\PB_i(c)'$ for values after the price reduction. Let $i \in H_k$. We need to prove that all items in $\xx_i$ are MPB items for $i$ after the price change. First we claim that $\MPB_i' = \alpha \MPB_i$. For $c \in \xx(\agents \setminus H_k)$, we have by choice of $\alpha$ that
	\[
		\PB_i'(c) = \PB_i(c) = \frac{\PB_i(c)}{\MPB_i} \MPB_i \ge \alpha \MPB_i.
	\]
	For all $c \in \xx(H_k)$,
	\[
		\PB_i'(c) = \alpha \PB_i(c) \ge \alpha \MPB_i.
	\]
	Finally for all $c \in \xx_i$ we have $\PB_i(c) = \MPB_i$ since $c$ was an MPB item for $i$ before the price change. Hence
	\[
		\PB_i'(c) = \alpha \PB_i(c) = \alpha \MPB_i.
	\]
	From these, it follows that indeed $\MPB_i' = \alpha \MPB_i$, and that $\PB_i'(c) = \MPB_i'$ for all $c \in \xx_i$. So all items owned by $i$ are MPB items for $i$ after the price change, as required.
\end{proof}
 
For the statement about termination, we need a few properties of Phase 2b of the algorithm, which is very similar to Phase 2 of the original algorithm due to \citet{barman2018finding}. 
The proof of this result is deferred to the appendix.

\begin{lemma}[Properties of Phase 2b]
	\label{lem:properties-phase2b}
	Consider a run of Phase 2b, and assume that $(\xx, \price)$ is an equilibrium at the start of the run.
	\begin{enumerate}
		\item \label{enum:phase2b-poly} The run terminates after $\text{poly}(n,m)$ time.
		\item \label{enum:phase2b-ls-spending} Least spending $\min_{i\in \agents} \price(\xx_i)$ never decreases during the run.
	\end{enumerate}
\end{lemma}

Assuming the induction hypotheses and using the lemmas mentioned above, we can now prove that the algorithm terminates, and is hence correct.
\begin{lemma}
	Assume that \hypdisjoint to \hypmpb hold for all $k \ge 1$ such that \Cref{alg:full-algo} reaches time $t_{k,a}$. Then the algorithm terminates in polynomial time and returns a pEF1 equilibrium.
\end{lemma}
\begin{proof}
	Every step of the algorithm is well-defined. This is obvious except for \cref{line:non-entitled-item-exists}, where the algorithm implicitly asserts the existence of a chore satisfying a certain property. But by \hyptransferable such a chore exists every time \cref{line:non-entitled-item-exists} is reached.
	
	The only way that the algorithm can terminate is if $(\xx, \price)$ is pEF1 (\cref{line:return}), at which time it is also an equilibrium by \Cref{lem:always-equilibrium}. So, it suffices to show that the algorithm terminate in polynomial time.

	Consider an execution of Phase 2a. For any value of $\ell$, consider the while loop in \cref{line:phase-2b-start}. In each step of the while loop, a chore is transferred from an agent in $H_\ell$ to an agent in $H_t$ for some $t > \ell$. Since chores only move from lower-numbered $H$-sets to higher-numbered $H$-sets, each item can be moved at most once. Hence, this while loop terminates after at most $m$ steps, and hence, Phase 2a terminates in polynomial time. 
	
	Phase 2b terminates in polynomial time by \Cref{lem:properties-phase2b}(a) which we can apply since $(\xx,\price)$ is an equilibrium by \Cref{lem:always-equilibrium}.
	
	Phase 3 can be executed at most $n$ times, because in each execution at least one agent (the least spender $\ls$) is assigned to a set $H_k$, and that agent was not previously assigned to such a set by \hypdisjoint. Since there are only $n$ agents, this can happen at most $n$ times.
\end{proof}

We now turn to proving our induction hypotheses. Recall that we prove them by induction on the iteration number $k$. First, let us prove them in the base of $k=1$. 
\begin{lemma}[Base case]
	\label{lem:base-case}
	\hypdisjoint to \hypmpb hold for $k = 1$.
\end{lemma}
\begin{proof}
	\hypdisjoint holds vacuously. \hyptransferable also holds vacuously because \cref{line:non-entitled-item-exists} is never reached in iteration 1, because the for-loop of Phase 2a is not executed. \hypnotviolator holds because otherwise Phase 2b would not have stopped. \hypentitled holds by the definition of $\entitled(i)$.
	
	Call a chore \emph{very difficult} if $v_i(c) = -p$ for all $i \in \agents$. In \cref{line:initialize-prices}, we set prices to be $p^2$ for very difficult chores, and $p$ for other chores.
	
	Consider time $t_{1,b}$, when prices are the same as at initialization. Note that for all $i \in \agents$, we have $\MPB_i = 1/p$, because we have assumed that every $i$ values at least one item $c$ at $-1$, so $c$ is not very difficult and $\price(c) = p$ giving $\PB_i(c) = 1/p$. (Clearly $\MPB_i$ cannot be less than $1/p$ since the only possible pain-per-buck ratios are $p/p$, $1/p$, and $p/p^2$. The ratio $1/p^2$ is not possible since only very difficult chores have price $p^2$.) Let $c$ be a very difficult item. Then $\PB_i(c) = p/p^2 = 1/p$ for all $i \in \agents$. Hence $c$ is an MPB chore for all agents. It follows that if $i$ is the owner of $c$ at time $t_{1,b}$, then $i \in H_1$. Thus, at $t_{1,b}$ all very difficult chores are owned by agents in $H_1$. Next, let $c \in \cup_{i' \in \agents \setminus H_1} \xx_{i'}$ be a chore not owned by an agent in $H_1$, say $c \in \xx_j$. Then for all $i \in H_1$ we must have $\PB_i(c) > \MPB_i$ or else $c$ is an MPB chore for $i$ and then we would have $j \in H_1$ by the definition of $H_1$. Hence $\PB_i(c) = 1$. It follows that in \cref{line:setting-alpha}, we set 
	\[
	\alpha = \min \left\{\PB_i(c)/\MPB_i : i \in H_1, c \in \xx({\agents \setminus H_1}) \right\} =  \tfrac{1}{1 / p} = p.
	\]
	This gives \hypalpha. 
	
	Next, in \cref{line:update-prices}, we multiply the price of each item owned by $H_1$ by $1/\alpha$. In particular, we update the price of every very difficult item from $p^2$ to $p$, and we may update some other chores' prices from $p$ to $1$. After this update at time $t_{1,a}$, we thus have $\price(c) \in \{1, p\}$ for all chores $c \in \items$. Also, any item $c$ that is now priced 1 must have had its price updated, so $c$ is owned by someone in $H_1$, and hence $c \in \entitled(i)$ for some $i \in H_1$. This gives \hypprices. 
	
	Finally, we calculate the values of $\MPB_i$ after the price change, i.e. at time $t_{1,a}$. Let $i \in H_1$. There exists some item $c$ with $v_i(c) = -1$. Before the price change, $c$ was an MPB chore for $i$. Thus by the definition of $H_1$, the owner of $c$ is in $H_1$. Now $c$'s price has changed from $p$ to $1$. Thus if $v_i(c) = -1$ then $\price(c) = 1$, and $\PB_i(c)=1$. On the other hand, for chores $c$ with $v_i(c) = -p$, we have $\price(c) \le p$ and so $\PB_i(c) \ge 1$.  It follows that $\MPB_i = 1$ after the price change.
	Next let $j \in \agents \setminus H_1$. Note that $j$ was not a least spender at $t_{1,b}$ (because least spenders are in $H_1$). Hence $\price(\xx_j) > 0$ and so $\xx_j \neq \emptyset$. Take some $c \in \xx_j$. At $t_{1,b}$, we had $\MPB_j = 1/p$, so $\PB_j(c) = 1/p$. The price of $c$ did not change, because $c$ is not owned by $H_1$. So also at $t_{1,a}$, we have $\PB_j(c) = 1/p$ and hence $\MPB_j = 1/p$ because $1/p$ is the smallest possible pain-per-buck ratio. This gives \hypmpb.
\end{proof}

From now on, we assume that \hypdisjoint to \hypmpb hold for all $\ell$ with $1 \le \ell \le k$ for some $k\ge 1$. Our goal is to show that \hypdisjoint to \hypmpb hold for iteration $k+1$. The next lemma shows that \hyptransferable holds for iteration $k+1$.
\begin{lemma}
	\label{lem:non-entitled-transferable}
	During iteration $k + 1$, each time the algorithm reaches \cref{line:non-entitled-item-exists}, there exists a chore $c \in \xx_i \setminus \entitled(i)$, and any such chore is an MPB chore for the agent $j$ identified in \cref{line:phase-2a-def-j}.
\end{lemma}
Before we prove \Cref{lem:non-entitled-transferable}, we need two additional results.
\begin{lemma}
	\label{lem:possession-subset}
	For all $1 \le \ell < k$, we have 
	$
	\xx^{t_{\ell, b}}(H_{\ell + 1} \cup \ldots \cup H_{k}) \subseteq \xx^{t_{k, b}} (H_{\ell + 1} \cup \ldots \cup H_{k}).
	$
\end{lemma}
\begin{proof}
Consider some agent $s \in H_r \subseteq H_{\ell + 1} \cup \ldots \cup H_{k}$ and some item $c \in \xx^{t_{\ell,b}}_s$.
Suppose that the price of item $c$ changes at some iteration $u$ with $\ell+1 \le u \le k$. Then $c \in \entitled(i)$ for some $i \in H_u$, and then by \hypentitled we have $c \in \xx^{t_{k,b}}_i$ as desired. Otherwise, the price of $c$ does not change. Consider the owner of $c$ at time $t_{r,b}$, say $c \in \xx^{t_{r,b}}_j$. Then $j \not\in H_r$ since the price of $c$ did not change. Now, at time $t_{\ell,b}$, agent $s$ owned $c$, so at that time $\PB_s(c) = \MPB_s$ (because the algorithm is always in equilibrium by \Cref{lem:always-equilibrium}).  Since the price of $c$ has not changed, and since the value of $\MPB_s$ has continued to be $1/p$ by \hypmpb applied to iteration $r-1$, we still have $\PB_s(c) = \MPB_s$ at time $t_{r,b}$. But then since $s \in H_r$, we also have $j \in H_r$, a contradiction.
\end{proof}

The next lemma is a sort of load balancing lemma. Intuitively, it says that if a group of agents are allocated some chores of equal price, and over the time, they receive more chores and the prices of the chores increase, then regardless of how the distribution of those chores between the agents changes, the minimum spending in the group can only increase. 
\begin{lemma}
	\label{lem:higher-prices-higher-ls}
	Let $\agents$ be a set of agents. Let $\items$ and $\items'$ be two sets of chores with $|\items| \le |\items'|$.
	Let $(\xx,\price)$ and $(\xx',\price')$ be pEF1 equilibria, where $\xx$ and $\xx'$ are allocations of $\items$ and $\items'$, respectively, to the agents in $\agents$. Suppose that $\price(c) = 1$ for all $c \in \items$ and that $\price'(c') \in \set{1, p}$ for all $c' \in \items'$. Then
	$
	\min_{i \in \agents} \price(\xx_i) \le \min_{i \in \agents} \price'(\xx'_{i}).
	$
\end{lemma}
\begin{proof}
	Let the least spenders of $\xx$ and $\xx'$ be $\ls_1$ and $\ls_2$ respectively. 
	Let $k = \floor{\nicefrac{\len{\items}}{n}}$.
	Because $\xx$ is pEF1 and all items are priced $1$, then each agent is allocated $k$ or $k + 1$ items in $\xx$ and so $\price(\xx_{\ls_1}) = k$. 
	
	For a contradiction, assume that $\price'(\xx'_{\ls_2}) < \price(\xx_{\ls_1}) = k$. Then $\len{\xx'_{\ls_2}} < k$. For other agents $o \in \agents \setminus \set{\ls_2}$, we have $\len{\xx'_o} - 1 \le \priceone'(\xx'_o) \le \price'(\xx'_{\ls_2}) < k$, where the first inequality holds because $\price'(c) \ge 1$ for all $c$, and the second holds because $\xx'$ is pEF1. Thus $\len{\xx'_o} - 1 < k$ and so $\len{\xx'_o} \le k$. Now,
	\[
	\len{\items'} = \len{\xx'_{\ls_2}} + \sum_{o \in \agents \setminus \set{ls_2}} \len{\xx'_o} \le k - 1 + (n - 1)k = nk - 1,
	\]
	However, $\len{\items'} \ge \len{\items} \ge nk$, which is a contradiction.
\end{proof}

We are now ready to prove \Cref{lem:non-entitled-transferable}.
\begin{proof}[Proof of \Cref{lem:non-entitled-transferable}]
	We prove the second part first. Any $c \in \xx_i \setminus \entitled(i)$ must have $\price(c) = p$ by \hypprices for iteration $k$. Since $j \in H_1 \cup \dots \cup H_k$, we have $\MPB_j = 1$ by \hypmpb applied to iteration $k$. Thus $\PB_j(c)$ cannot be $1/p$, so $\PB_j(c) = 1 = \MPB_j$. Hence $c$ is an MPB chore for $j$.
	
	For the first part, assume that there is some time $t$ when the algorithm select agents $i \in H_\ell$ and $j \in H_{\ell+1}\cup \dots \cup H_k$ where no $c \in \xx_i^t \setminus \entitled(i)$ exists. Let $t$ be the first such time. By \cref{line:phase-2a-violator-condition}, we have $\priceone^t(\xx_i^t) > \price^t(\xx_j^t)$. 
By \hypentitled, we had $\entitled(i) \subseteq \xx_i^{t_{k},b}$. Since $t$ is the time of first failure, so far no entitled item has been transferred away from $i$. Thus $\xx_i^t = \entitled(i)$. By definition of $\entitled(i)$ and since $i \in H_\ell$, thus $\xx_i^t = \xx_i^{t_{\ell,a}}$. Now we have
	\begin{align*}
		\priceone^t(\xx_i^t) &= \priceone^{t_{\ell,a}}(\xx^{t_{\ell,a}}_i) \tag{since prices of entitled chores have not changed after $t_{\ell,a}$ by \hypdisjoint} \\
		&= \tfrac1p \cdot \priceone^{t_{\ell,b}}(\xx^{t_{\ell,b}}_i) \tag{since $\alpha = \frac1p$ in iteration $\ell$ by \hypalpha} \\
		&\le \tfrac1p \cdot \text{min}_{o \in \agents} \: \price^{t_{\ell,b}}(\xx^{t_{\ell,b}}_o). \tag{since $i \in H_\ell$ was not a violator at $t_{\ell,b}$ by \hypnotviolator}
	\end{align*}
	By \Cref{lem:possession-subset},
	$
	\xx^{t_{\ell, b}}(H_{\ell + 1} \cup \ldots \cup H_{k}) \subseteq \xx^{t_{k, b}} (H_{\ell + 1} \cup \ldots \cup H_{k}).
	$
	And, so far in Phase 2a of iteration $k+1$, no item has been transferred out of $H_{\ell + 1} \cup \ldots \cup H_{k}$. Therefore, 
	$
	\xx^{t_{k, b}} (H_{\ell + 1} \cup \ldots \cup H_{k}) \subseteq \xx^{t} (H_{\ell + 1} \cup \ldots \cup H_{k}).
	$
	By \hypprices, at $t_{\ell, b}$, all chores owned by $H_{\ell + 1} \cup \ldots \cup H_{k}$ were priced $p$. Now, at $t$, they own a superset of those chores with prices of $1$ or $p$. By applying \Cref{lem:higher-prices-higher-ls} with the first set of chores being $\xx^{t_{\ell, b}}(H_{\ell + 1} \cup \ldots \cup H_{k})$ all priced 1, and the second set being $\xx^{t'} (H_{\ell + 1} \cup \ldots \cup H_{k})$ with their current prices at time $t'$, we can conclude that
	\[
	\tfrac1p \cdot  \min_{o \in H_{\ell + 1} \cup \ldots \cup H_{k}} \price^{t_{\ell, b}}(\xx_o^{t_{\ell, b}}) \le \min_{o \in H_{\ell + 1} \cup \ldots \cup H_{k}} \price^{t}(\xx_o^t) = \price^{t} (\xx_j^t),
	\]
	where the last equality holds by choice of $j$.
	Combining this with the previous inequalities, we get $\priceone^t(\xx_i^t) \le \price^{t} (\xx_j^t)$, a contradiction.
\end{proof} 
The next lemma proves the usefulness of Phase 2a, which is that at the end of this phase, no agent in $H_1 \cup \dots \cup H_k$ (i.e., no agent who has gone through a price reduction) can be a violator. 
\newcommand{\tmid}{t_{\text{mid}}}
\begin{lemma}
	\label{lemma:phase2a-no-violators}
	Let $\tmid$ denote the time when the algorithm reaches \cref{line:phase-2b-start} in iteration $k + 1$, i.e. when Phase 2a ends and Phase 2b begins. We claim that at time $\tmid$, no agent in $H_1 \cup \dots \cup H_k$ is a violator.
\end{lemma}
\begin{proof}
	For an allocation $\xx$ and sets $S,T \subseteq \agents$, let us write
	\[
	\noviol(\xx, S \to T) \iff \priceone(\xx_i) \le \price(\xx_j) \text{ for all $i \in S$ and $j \in T$}.
	\]
	In this notation, an agent $i$ is not a violator if and only if $\noviol(\xx, \{i\} \to \agents)$.
	
	For each $\ell \in [k]$, write $R_{\ell} = H_\ell \cup \dots \cup H_k$. By \hypnotviolator applied to iteration $k$, no agent in $R_\ell$ is a violator at time $t_{k,b}$. In particular, this means that at time $t_{k,b}$ we have $\noviol(\xx, R_{\ell} \to \agents \setminus R_{\ell})$ for each $\ell \in [k]$. The same is true at time $t_{k,a}$, because the price reduction reduces the prices of goods held by $R_\ell$ but not those held by $\agents\setminus R_\ell$. So we have
	\begin{equation}
		\label{eq:no-viol-rl-base}
		\noviol(\xx, R_{\ell} \to \agents \setminus R_{\ell}) \quad \text{for all $\ell\in[k]$, at time $t_{k,a}$.}
	\end{equation}

	\paragraph{Induction on $\ell$} Now, we prove inductively for $\ell = k-1, \dots, 1$, that at the end of the for-loop iteration corresponding to $\ell$, we have
	\begin{equation}
		\label{eq:no-viol-rl-induction-hypothesis}
		\noviol(\xx^{\ell}, R_{\ell} \to \agents),
\end{equation}
	where $\xx^\ell$ is the allocation at the end of the $\ell$-th iteration of the for-loop of Phase 2a.
	
	As a base case, we take $\ell = k$, noting that just before the for-loop starts (i.e. at time $t_{k,a}$), we have $\noviol(\xx^{k}, R_{k} \to \agents)$ because
	\begin{itemize}
		\item from \eqref{eq:no-viol-rl-base} we know that $\noviol(\xx^{k}, R_{k} \to \agents \setminus R_k)$, and
		\item from \hypnotviolator we had $\noviol(\xx^{k}, H_{k} \to H_k)$ at time $t_{k,b}$, and since the price reduction changes the prices of all chores held by $H_k$ by the same factor $\alpha$, we also have $\noviol(\xx^{k}, H_{k} \to H_k)$ at time $t_{k,a}$.  Since $R_k = H_k$, hence $\noviol(\xx^{k}, R_{k} \to R_k)$.
	\end{itemize}
	
	Once we have established the induction step, we can apply \eqref{eq:no-viol-rl-induction-hypothesis} for $\ell = 1$ to get that $R_1 = H_1 \cup \dots \cup H_k$ are not violators at $\tmid$, which is the lemma statement. \\
	
	Suppose that \eqref{eq:no-viol-rl-induction-hypothesis} holds for $\ell$ at the start of iteration $\ell - 1$. (This is true either by the base case, or because \eqref{eq:no-viol-rl-induction-hypothesis} for $\ell$ held at the end of iteration $\ell$ and hence also at the start of iteration $\ell - 1$.) Our goal is to show that \eqref{eq:no-viol-rl-induction-hypothesis} holds for $\ell-1$ at the end of iteration $\ell -1$. We prove this in two steps: first, we show $\noviol(\xx^{\ell-1},R_{\ell-1} \to \agents \setminus R_{\ell-1})$, and then we show $\noviol(\xx^{\ell-1},R_{\ell-1} \to R_{\ell-1})$.
	
	\paragraph{No violators to $\agents \setminus R_{\ell - 1}$.} For all agents $s \in \agents \setminus R_{\ell - 1}$, at the start of iteration $\ell - 1$ we have
	\begin{itemize}
		\item $\max_{i \in R_{\ell}} \priceone(\xx^\ell_i) \le \price(\xx^\ell_s)$ by the induction hypothesis \eqref{eq:no-viol-rl-induction-hypothesis}, and
		\item $\max_{i \in H_{\ell - 1}} \priceone(\xx^\ell_i) \le \price(\xx^\ell_s)$ by \eqref{eq:no-viol-rl-base}, noting that the bundles of $H_{\ell-1}$ and of $\agents \setminus R_{\ell - 1}$ have not been changed since $t_{k,a}$.
	\end{itemize}
	Hence, $\max_{i \in R_{\ell - 1}} \priceone(\xx^\ell_i) \le \price(\xx^\ell_s)$. 
	
	During the execution of iteration $\ell - 1$ of Phase 2a, the value $\max_{i\in R_{\ell - 1}}\priceone(\xx_i)$ never increases: This is because if we transfer $c$ from agent $i$ to agent $j$, then 
	\begin{itemize}
		\item $\priceone(\xx_i)$ decreases because $i$ gave away an item,
		\item $\priceone(\xx_j)$ increases but not too much: Write $\xx$ and $\xx'$ for the allocations before and after the transfer, and recall that $\priceone(\xx_i) > \price(\xx_j)$ since we performed the transfer. Then
		\[
		\priceone(\xx_j') = \priceone(\xx_j \cup \{c\}) \le \price(\xx_j) < \priceone(\xx_i).
		\]
		Hence $\priceone(\xx_j')$ is smaller than the previous maximum value of $\priceone$.
	\end{itemize}
	Recall that $\xx^{\ell-1}$ is the allocation at the end of iteration $\ell - 1$. Thus for all $s \in \agents \setminus R_{\ell - 1}$,
	\[
	\max_{i\in R_{\ell - 1}}\priceone(\xx_i^{\ell - 1}) \le \max_{i\in R_{\ell - 1}}\priceone(\xx_i^{\ell}) \le \price(\xx_s^{\ell}) =\price(\xx_s^{\ell - 1})
	\]
	where the last equality holds because the bundle $\xx_s$ has not changed since the start of iteration $\ell - 1$. Therefore, $\noviol(\xx^{\ell - 1}, R_{\ell - 1} \to \agents \setminus R_{\ell - 1})$. 
	
	\paragraph{No violators for $ R_{\ell - 1}$.}
	
At the start of iteration $\ell - 1$ we have all of the following:
	\begin{enumerate}
		\item[(a)] $\noviol(\xx^\ell, H_{\ell - 1} \to H_{\ell - 1})$ by \hypnotviolator since the bundles of $H_{\ell-1}$ have not changed since $t_{k, b}$,
		\item[(b)] $\noviol(\xx^\ell, R_{\ell } \to R_{\ell})$ by inductive hypothesis \eqref{eq:no-viol-rl-induction-hypothesis},
		\item[(c)] $\noviol(\xx^\ell, R_\ell \to H_{\ell - 1})$ by inductive hypothesis \eqref{eq:no-viol-rl-induction-hypothesis}.
	\end{enumerate}
	We now show inductively that after each transfer, we still have (a), (b), and (c).
	
	So suppose (a), (b), and (c) hold for allocation $\xx$. We now transfer item $c$ from $i \in H_{\ell - 1}$ to $j \in R_\ell$, obtaining allocation $\xx'$. We show that (a), (b), and (c) also hold for allocation $\xx'$.
	
	\begin{enumerate}
		\item[(a)] For all $s \in H_{\ell - 1}$,
		\[
		\priceone(\xx'_s) \le \priceone(\xx_s) \le \priceone(\xx_i) \le \price(\xx_i) - \price(c) = \price(\xx'_i),
		\]
		where the first inequality holds because $H_{\ell-1}$ did not receive items, and the second inequality holds by choice of $i$. Hence, $\noviol(\xx', H_{\ell - 1} \to \set{i})$. Because (a) held before the transfer, because the transfer did not change the bundles for $H_{\ell - 1} \setminus \{i\}$, and because $\priceone(\xx_i') \le \priceone(\xx_i)$, we have $\noviol(\xx', H_{\ell - 1} \to H_{\ell - 1} \setminus \{i\})$. Putting the two together, we have $\noviol(\xx', H_{\ell - 1} \to H_{\ell - 1})$.
		\item[(b)] For all $s \in R_{\ell}$,
		\[
		\priceone(\xx'_j) \le \price(\xx_j \cup \set{c}) - \price(c) = \price(\xx_j) \le \price(\xx_s) \le \price(\xx'_s),
		\]
		where the penultimate inequality holds by choice of $j$ and the last because $R_\ell$ does not give away items.
		Therefore, $\noviol(\xx', \set{j} \to R_{\ell})$.  Because (b) held before the transfer, because the allocation did not change for $R_{\ell} \setminus \{j\}$, and because $\price(\xx'_j) \ge \price(\xx_j)$, we have $\noviol(\xx', R_{\ell} \setminus \{j\} \to R_{\ell})$. Putting the two together, we have $\noviol(\xx', R_{\ell} \to R_{\ell})$.
		\item[(c)] 
		Note first that since $c$ is not an entitled chore, we have $\price(c) = p$ due to \hypprices, so $c$ is a chore with maximum price. Therefore
		\begin{equation}
			\label{eq:max-price}
			\priceone(\xx'_j) = \price(\xx_j \cup \{c\}) - \price(c) = \price(\xx_j) 
			\quad \text{and} \quad
			\priceone(\xx_i) = \price(\xx_i) - \price(c) = \price(\xx_i'). 
		\end{equation}
		We show that (c) holds for $\xx'$ in four parts.
		\begin{itemize}
			\item Because (c) held before the transfer, we have $\noviol(\xx, R_{\ell} \setminus \set{j} \to H_{\ell - 1} \setminus \set{i})$. Since the transfer did not change the bundles of $R_\ell \setminus \set{j}$ and of $H_{\ell - 1} \setminus \set{i}$, we thus have $\noviol(\xx', R_{\ell} \setminus \set{j} \to H_{\ell - 1} \setminus \set{i})$.
			\item Using \eqref{eq:max-price}, and because we performed the transfer, we have
			\begin{equation}
				\label{eq:j-does-not-violate-i}
				\priceone(\xx'_j) = \price(\xx_j) < \priceone(\xx_i) = \price(\xx_i').
			\end{equation}
			Hence $\noviol(\xx', \set{j} \to \set{i})$.
			\item For all $s \in H_{\ell - 1} \setminus \set{i}$, we have
			\[
			\priceone(\xx'_j) \overset{\eqref{eq:j-does-not-violate-i}}{<} \price(\xx'_i) \overset{\eqref{eq:max-price}}{=} \priceone(\xx_i) \le \price(\xx_s) = \price(\xx'_s)
			\]
			which shows $\noviol(\xx', \set{j} \to H_{\ell - 1} \setminus \set{i})$.
			\item For all $s \in R_{\ell} \setminus \set{j}$, we have
			\[
			\priceone(\xx'_s) = \priceone(\xx_s) \le \price(\xx_j) \overset{\eqref{eq:max-price}}{=}  \priceone(\xx'_j) \overset{\eqref{eq:j-does-not-violate-i}}{<} \price(\xx'_i)
			\]
			which shows $\noviol(\xx', R_i \setminus \set{j} \to \set{i})$.
		\end{itemize}
		Putting all these together, we get $\noviol({\xx', R_{\ell} \to H_{\ell - 1}})$, which is (c). 
	\end{enumerate}
	The induction tells us that (a), (b), and (c) hold when iteration $\ell-1$ finishes, i.e. they hold for $\xx^{\ell-1}$. In addition, because the iteration ended (\cref{line:phase-2a-violator-condition}), we must have $\noviol(\xx^{\ell - 1} , H_{\ell - 1}  \to R_{\ell})$. This together with (a), (b), and (c) gives $\noviol(\xx^{\ell - 1} , R_{\ell - 1}  \to R_{\ell - 1})$, as desired.
\end{proof}
 
The next lemma proves a useful guarantee for Phase 2b.
\begin{lemma}
	\label{lem:phase-2b-analysis}
	During the execution of Phase 2b in iteration $k + 1$, no entitled items are transferred. Further, at the end of Phase 2b, no agent $i \in H_1 \cup \cdots \cup H_k$ is a violator.
\end{lemma}
\begin{proof}
	Write $R = H_1 \cup \dots \cup H_k$.
	Recall that at the start of Phase 2b, no agent in $R$ was a violator (\Cref{lemma:phase2a-no-violators}). Thus, for $i\in R$ to give away a chore in Phase 2b, $i$ needs to become a violator. This can only happen if $i$ receives a chore during Phase 2b. 
	
	Consider a transfer during Phase 2b of a chore $c$ from $s \not\in R$ to $i \in R$ at time $t$, after which $i$ becomes a violator. At time $t$, there was an MPB alternating path $\ls \reachable \smash{i \stackrel{c}{\gets} s}$ of length $\ell + 1$, say. (At this time, $i$ cannot be a least spender, because a least spender cannot become a violator after being given one item.) Since Phase 2b chooses MPB alternating paths of minimum length, there was no suitable path of length $\ell$ available. At time $t'$ (the time step immediately after the transfer), there is an MPB alternating path $\ls \reachable i$ of length $\ell$, which means that $i$ is the violator uniquely closest to $\ls$ at this point. It follows that at $t'$, we perform a transfer of some item $c'$ from $i$ to some agent $j$, using an MPB alternating path $\ls \reachable \smash{j \stackrel{c'}{\gets} i}$ of length $\ell$. We now claim that $j \not\in R$. This is because if $j$ was a member of $R$, then at time $t$, we would have $\smash{j \stackrel{c}{\gets} s}$ (because $\price(c) = p$ by \hypprices and thus $\PB_j(c) = 1 = \MPB_j$). Thus $\ls \reachable \smash{j \stackrel{c}{\gets} s}$ would have been an MPB alternating path of length $\ell$ to violator $s$ at time $t$, contradicting that the shortest such path had length $\ell + 1$. Hence $j \not \in R$. Now, because we performed the transfer $\smash{j \stackrel{c'}{\gets} i}$, item $c'$ is an MPB chore for $j$. Thus from \hypmpb, $\PB_j(c') = 1/p$. It follows that $\price(c') = p$ and $v_j(c') = -1$. But then $c' \not\in \entitled(i)$, because the only entitled chores with price $p$ are very difficult chores that every agent values at $-p$ (see the proof of \Cref{lem:base-case}), and $v_j(c') \neq -p$ so $c'$ is not a very difficult item. Thus $c' \not\in \entitled(i)$. Recall that $t$ is the time before the transfer of $c$ from $s$ to $i$, that $t'$ is the time after that transfer, and let $t''$ be the time after the transfer of $c'$ from $i$ to $j$. (So $\xx_i^{t''} = \xx_i^{t} \cup \{c\} \setminus \{c'\}$.) We now show that agent $i$ is not a violator anymore at time $t''$. First, note that $\xx_i^{t''}$ is obtained from $\xx_i^{t}$ by adding an item priced $p$ and removing an item priced $p$. Thus $\priceone(\xx_i^{t''}) = \priceone(\xx_i^{t})$. Since $i$ was not a violator at time $t$ (so $\priceone(\xx_i^{t}) \le \min_{s\in\agents} \price(\xx_s)$), and least spending cannot have decreased since time $t$ by \Cref{lem:properties-phase2b}(2), it follows that $i$ is not a violator at time $t''$.
	
	Thus we have shown that if $i \in R$ becomes a violator due to being given an item from some agent $s \not\in R$, then $i$ ceases to be a violator in the immediately next step by giving away a non-entitled item to an agent $j \not\in R$. The only other way that $i \in R$ could give away an item is if $i$ becomes a violator due to being given an item from some agent $s \in R$, but we show this never happens: if it did, consider the first time some $i \in R$ becomes a violator in this way. But then $s \in R$ must have previously become a violator due to being given an item from some $s'\not\in R$. But we have already proven that this is a contradiction, because $s'$ will immediately become a non-violator without giving an item to a member of $R$.
\end{proof} 
Finally, we prove the induction step of our induction hypotheses.
\begin{lemma}
	Suppose \Cref{alg:full-algo} reaches time $t_{k+1, a}$. Then \hypdisjoint to \hypmpb hold for $k+1$.
\end{lemma}
\begin{proof}
	Write $R = H_1 \cup \dots \cup H_k$.
We have already proved \hyptransferable in \Cref{lem:non-entitled-transferable}. 
For \hypnotviolator, by \Cref{lem:phase-2b-analysis} no agent in $R$ is a violator at the end of Phase 2b, and thus at time $t_{k+1,b}$, which is \hypnotviolator.
	
	For \hypdisjoint, note that because the algorithm has reached time $t_{k+1,b}$, it did not return an allocation in  \cref{line:return}. Thus there exists a violator $s$. As we just proved, $s \not \in R$. By inductive hypothesis \hypprices, at time $t_{k,a}$, we had $\price(c) = p$ for all $c \in \xx_s^{t_{k,a}}$. Because $s \not\in R$, the bundle $\xx_s$ is not changed during Phase 2a of iteration $k + 1$. By \Cref{lem:phase-2b-analysis}, during Phase 2b, no entitled items are ever transferred. Thus using inductive hypothesis \hypprices, during Phase 2b only items priced $p$ are transferred. So all chores owned by $s$ at $t_{k+1, b}$ are priced $p$. Note that $\xx_s \neq \emptyset$, because we had $\xx_s \neq \emptyset$ at $t_{1,a}$ because $i \not\in H_1$ (see the proof of \Cref{lem:base-case}) and Phases 2a and 2b never take away an agent's last item. So select some chore $c\in\xx_s$. We show that $R \cap H_{k+1} = \emptyset$. Let $i \in R$. We wish to prove that $i \not\in H_{k+1}$.  By inductive hypothesis \hypmpb, we had $\MPB_i = 1$ at time $t_{k,a}$ and hence also at $t_{k+1,b}$. Then $\PB_i(c) = 1$ because $\price(c) = p$ (and $\PB_i(c) = 1/p$ would contradict $\MPB_i = 1$) and thus $c$ is an MPB item for $i$.  If we had $i \in H_{k+1}$, then (due to item $c$) also $s \in H_{k+1}$ by definition of $H_{k+1}$. But this is a contradiction because violators cannot be in $H_{k+1}$ since Phase 2b then would not have terminated. Hence $i \not\in H_{k+1}$. Thus $R \cap H_{k+1} = \emptyset$, which is \hypdisjoint.
	
	For \hypentitled, note first that at time $t_{k+1,a}$, for all $i \in H_{k+1}$, the algorithm has just set $\entitled(i) = \xx_i$ and thus $\entitled(i) \subseteq \xx_i$. For $i \in R$, note that we had $\entitled(i) \subseteq \xx_i$ at time $t_{k,a}$ by inductive hypothesis. Phase 2a never transfers an entitled item of $i$. By \Cref{lem:phase-2b-analysis}, Phase 2b does not do this either. So we have proven \hypentitled.
	
	Next, check \hypalpha. We have shown that in iteration $k + 1$, no entitled chores were transferred. Thus all chores that were transferred  had price $p$ at time $t_{k,a}$ (and thus also at time $t_{k+1, b}$) by inductive hypothesis \hypprices. We now prove \hypalpha, that $\alpha = p$ in iteration $k + 1$. Let $c \in \xx^{t_{k+1, b}}(\agents\setminus (H_1 \cup \dots \cup H_{k+1}))$ be a chore not owned by $R \cup H_{k+1}$ at time $t_{k+1,b}$; say it is owned by $s$. (Such a chore exists because otherwise the algorithm would have terminated at \cref{line:return}.) We have $\price(c) = p$. Item $c$ cannot be an MPB chore for anyone in $H_{k+1}$ since otherwise $s \in H_{k+1}$, contradiction. Since $\MPB_i = 1/p$ for all $i \in H_{k+1}$, we must have $\PB_i(c) = 1$. Hence we have $\alpha = 1 / (1/p) = p$, giving \hypalpha.
	
	For \hypprices, we only need to consider chores owned by $H_{k+1}$, since no other chores have their price changed. Let $c \in \xx^{t_{k+1, b}}_i$ for some $i \in H_{k+1}$. Since $\price^{t_{k+1, b}}(c) \in \{1,p\}$ by inductive hypothesis \hypprices and $\MPB_i = 1/p$ at time $t_{k+1, b}$ by \hypmpb, we have $\price^{t_{k+1, b}}(c) = p$. Since $\alpha  = p$ as we have just shown, then $\price^{t_{k+1, a}}(c) = \frac1\alpha p = 1$. Because $c \in \entitled(i)$, this gives \hypprices.
	
	Finally for \hypmpb, since chores owned by $H_{k+1}$ changed price from $p$ to $1$, for each $i \in H_{k+1}$ the value of $\MPB_i$ changes from $1/p$ to $1$. 
	For other agents, the MPB values have not changed: For $i \in \agents\setminus (H_1 \cup \dots \cup H_{k+1})$ we had $\MPB_i = 1/p$ at time $t_{k,a}$ by inductive hypothesis (H7). 
	At time $t_{k+1,b}$, agent $i$ owns at least one item $c$ (since $i \not \in H_1$) and since the algorithm always stays in equilibrium, $\PB_i(c) = 1/p$. 
	Since the price of $c$ was not changed in iteration $k+1$, we also have $\PB_i(c) = 1/p$ at time $t_{k+1,a}$. Hence at this time, $\MPB_i \le 1/p$, but $1/p$ is the smallest possible value, so $\MPB_i = 1/p$. For an agent $i \in R$, we had $\MPB_i = 1$ at time $t_{k,a}$ by \hypmpb. Since price reductions can only increase the value of $\MPB_i$ and $1$ is the highest possible pain-per-buck ratio, we also have $\MPB_i = 1$ at time $t_{k+1, a}$. This proves \hypmpb.
\end{proof}
 \section{MMS Under Restricted Utilities}\label{sec:MMS}

As discussed earlier, MMS allocations are not guaranteed to exist for arbitrary additive utilities. Prior work on allocating goods establishes that they always exist for binary utilities~\cite{BL16} and strictly lexicographic utilities~\cite{HSVX21}. In this section, we generalize these results to the classes of weakly lexicographic and factored personalized bivalued utilities. The following theorem summarizes our main result of this section.

\begin{theorem}\label{thm:mms}
	In every goods or chore division instance with weakly lexicographic or factored personalized bivalued utilities, an MMS allocation always exists and can be computed in polynomial time. 
\end{theorem}

\subsection{Ordered Instances and Valid Reductions}
Let us begin by reviewing two basic techniques which are commonly used in the literature on computing MMS allocations. Throughout this section, we let $\agents = [n]$ and $\items = [m]$.

\subsubsection{Ordered Instances}
\citet[Prop.~14]{BL16} show that when dealing with MMS allocations, one can assume, without loss of generality, that all agents have the same preference ranking over the items. This result was originally stated for goods, but the same proof works for  chores as well. 

\begin{lemma}[\cite{BL16}]\label{lem:ordered}
	Let $I = (\agents, \items, \vv)$ be a goods or chore division instance. Let $I' = (\agents,\items, \vv')$ be an \emph{ordered instance}  where, for each $i \in \agents$, $v'_i$ is a permutation of $v_i$ such that $|v'_i(1)| \ge \ldots \ge |v'_i(m)|$. If $\xx'$ is an MMS allocation for $I'$, then there exists an MMS allocation $\xx$ for $I$. Given $\xx'$, one can compute $\xx$ in polynomial time. 
\end{lemma}

Given \Cref{lem:ordered}, we will assume that all instances in this section (except in \Cref{sec:mms-po}, where our goal is to achieve PO in conjunction with MMS) are ordered instances. Specifically, we will assume that $|v_i(1)| \ge \ldots \ge |v_i(m)|$ for each agent $i \in \agents$. 

\subsubsection{Valid Reductions}
Another common idea used in the literature on finding (approximate) MMS allocations is that of \emph{valid reductions}~\cite{kurokawa2018fair,kurokawa2016can,amanatidis2017approximation,garg2019approximating,ghodsi2018fair,garg2021improved}. 

\begin{definition}[Valid Reduction]
	\label{def:valid-reduction}
	Let $I = (\agents, \items, \vv)$ be a goods or chore division instance, $i \in \agents$ be an agent, and $S \subseteq \items$ be a subset of items. The pair $(i,S)$ is a \emph{valid reduction} if
	\begin{enumerate}
		\item \label{enum:reduction-i-is-fine} $v_i(S) \ge \MMS^{n}_i(\items)$, and
		\item \label{enum:reduction-others-are-fine} $\MMS^{n-1}_j(\items \setminus S) \ge \MMS^n_j(\items)$ for all $j \in \agents \setminus \set{i}$.
	\end{enumerate}
\end{definition}

If $(i, S)$ is a valid reduction, we can allocate bundle $S$ to agent $i$, and ignore $i$ and $S$ subsequently. Formally, consider the reduced instance $I' = (\agents \setminus \{i\}, \items \setminus S, \vv)$ obtained from $I$ by removing $i$ and $S$. Then if $\xx'$ is an MMS allocation for $I'$, then the allocation $\xx$ with $\xx_i = S$ and $\xx_j = \xx'_j$ for all $j \neq i$ is an MMS allocation for $I$. This holds because agent $i$ receives her MMS value in $\xx$ by (\ref{enum:reduction-i-is-fine}), and for any other agent $j$, $v_j(\xx'_j) \ge \MMS^{n-1}_j(\items \setminus S) \ge \MMS^n_j(\items)$ by (2). 

Our proofs for both goods and chore division under both weakly lexicographic and factored personalized bivalued utilities work in the same fashion: we show that every instance admits a valid reduction which can be computed efficiently. The next lemma identifies one of the ways of finding a valid reduction. 

\begin{lemma}
	\label{lem:bag-reduction}
	For a goods or chore division instance $I = (\agents, \items, \vv)$, the pair $(i, S)$, where $i \in \agents$ and $S \subseteq \items$ is a valid reduction if $v_i(S) \ge \MMS^n_i(\items)$ and,
	for all agents $i' \in \agents \setminus \set{i}$, there is a maximin $n$-partition $P_{i'} = (S'_1,\ldots,S'_n)$ of agent $i'$ and a bundle $S' \in P_{i'}$ with 
	$S \subseteq S'$ for goods division and $S \supseteq S'$ for chore division.
\end{lemma}
\begin{proof}
	Fix an agent $i' \in \agents \setminus \set{i}$. Given the definition of a valid reduction, we only need to show that $\MMS_{i'}^{n-1}(\items \setminus S) \ge \MMS_{i'}^n(\items)$. Fix a maximin $n$-partition $P_{i'}$ as in the statement of the lemma, and note that $\MMS_{i'}^n(\items) = \min_{k \in [n]} v_{i'}(S'_k)$. Then, it is sufficient to construct an $(n-1)$-partition of $\items \setminus S$ such that each bundle in this partition is worth at least $\MMS_{i'}^n(\items)$ to agent $i'$. 
	
	\paragraph{Goods division:} In this case, we construct the desired $(n-1)$-partition of $\items \setminus S$ by starting from the maximin $n$-partition $P'$ of $\items$ and deleting bundle $S'$. Note that none of the other bundles contain a good from $S$ because $S \subseteq S'$. Since deleting a bundle can only improve the utility of the agent for the worst bundle, the utility of agent $i'$ for the worst bundle in the new partition is at least her utility for the worst bundle in $P'$, which is $\MMS_{i'}^n(\items)$. 
	
	\paragraph{Chore division:} In this case, we construct the desired $(n-1)$-partition of $\items \setminus S$ as before, by starting from the maximin $n$-partition $P'$ of $\items$, deleting bundle $S'$, and deleting any chores in $S \setminus S'$ from the remaining bundles. Since $S \supseteq S'$, this is an $(n-1)$-partition of $\items \setminus S$. Further, deleting a bundle and deleting some chores from the remaining bundles can only improve the utility of the agent for the worst bundle. Hence, the utility of agent $i'$ for the worst bundle in the new partition is at least her utility for the worst bundle in $P'$, which is $\MMS_{i'}^n(\items)$. 
\end{proof}

\subsection{Exact MMS Value for Factored Utilities}

Note that in order to check the validity of a reduction $(i,S)$ via \Cref{lem:bag-reduction}, we need to relate $S$ to one of the bundles in a maximin $n$-partition of every agent other than $i$. For this, we need to reason about what a maximin $n$-partition looks like for an agent. We show that in any goods or chore division instance with factored utilities (which covers weakly lexicographic and personalized factored bivalued utilities as special cases), a maximin $n$-partition of an agent (and hence, her MMS value) can be computed efficiently. This is in sharp contrast to the case of general additive utilities, for which the problem is known to be NP-hard for both goods and chores~\citep[p.~224, \textsc{3-Partition}]{GJ79}.

\Cref{alg:mms-value} considers the items in a nonincreasing order of their absolute value and greedily assigns them to the bundle with the lowest total absolute value. In the end, the value of the least-valued bundle is the MMS value. This works for both goods division and chore division. 
\Cref{fig:mms-value-alg-example} shows the execution of the algorithm on an example.

\begin{algorithm}
	\caption{Compute a maximin $n$-partition for a factored utility function $v$}
	\label{alg:mms-value}
	\SetAlgoLined
	\DontPrintSemicolon
	\SetAlgoNoEnd
	$\xx \leftarrow (\xx_i = \emptyset)_{i \in \agents}$ \tcp*{$\xx$ denotes a partial allocation}
	\For{$r \in \items$ \text{\emph{in a nonincreasing order of}} $\abs{v(r)}$}{
		$k^* \leftarrow \argmin_{k \in \agents} \abs{v(\xx_k)}$\;
		$\xx_{k^*} \leftarrow \xx_{k^*} \cup \set{r}$\;
	}
	\Return $\xx$
\end{algorithm}

\begin{figure}
	\centering
	\setlength{\tabcolsep}{0mm}
	\begin{tabular}{ccccccccc}
\begin{tikzpicture}
			[
			item-box/.style={anchor=south,minimum width=0.28cm, font=\footnotesize, draw, inner sep=0cm},
			newest-box/.style={fill=black!10},
			bag-box/.style={anchor=north,font=\footnotesize, inner xsep=0.3mm}
			]
			\matrix[anchor=south, column sep=0.0mm]
			{
				\draw node[item-box,newest-box,minimum height=3cm] {12};
				&
				\draw (0,0) node[item-box,minimum height=0cm] {};
				&
				\draw (0,0) node[item-box,minimum height=0cm] {};
				&
				\draw (0,0) node[item-box,minimum height=0cm] {};
				\\
				\node[bag-box] {$\xx_1$};
				&
				\node[bag-box] {$\xx_2$};
				&
				\node[bag-box] {$\xx_3$};
				&
				\node[bag-box] {$\xx_4$};
				\\
			};
			\node[font=\footnotesize] {Step 1};
		\end{tikzpicture}
		&
\begin{tikzpicture}
			[
			item-box/.style={anchor=south,minimum width=0.28cm, font=\footnotesize, draw, inner sep=0cm},
			newest-box/.style={fill=black!10},
			bag-box/.style={anchor=north,font=\footnotesize, inner xsep=0.3mm}
			]
			\matrix[anchor=south, column sep=0mm]
			{
				\draw node[item-box,minimum height=3cm] {12};
				&
				\draw (0,0) node[item-box,newest-box,minimum height=1.5cm] {6};
				&
				\draw (0,0) node[item-box,minimum height=0cm] {};
				&
				\draw (0,0) node[item-box,minimum height=0cm] {};
				\\
				\node[bag-box] {$\xx_1$};
				&
				\node[bag-box] {$\xx_2$};
				&
				\node[bag-box] {$\xx_3$};
				&
				\node[bag-box] {$\xx_4$};
				\\
			};
			\node[font=\footnotesize] {Step 2};
		\end{tikzpicture}
		&
\begin{tikzpicture}
			[
			item-box/.style={anchor=south,minimum width=0.28cm, font=\footnotesize, draw, inner sep=0cm},
			newest-box/.style={fill=black!10},
			bag-box/.style={anchor=north,font=\footnotesize, inner xsep=0.3mm}
			]
			\matrix[anchor=south, column sep=0mm]
			{
				\draw node[item-box,minimum height=3cm] {12};
				&
				\draw (0,0) node[item-box,minimum height=1.5cm] {6};
				&
				\draw (0,0) node[item-box,newest-box,minimum height=1.5cm] {6};
				&
				\draw (0,0) node[item-box,minimum height=0cm] {};
				\\
				\node[bag-box] {$\xx_1$};
				&
				\node[bag-box] {$\xx_2$};
				&
				\node[bag-box] {$\xx_3$};
				&
				\node[bag-box] {$\xx_4$};
				\\
			};
			\node[font=\footnotesize] {Step 3};
		\end{tikzpicture}
		&
\begin{tikzpicture}
			[
			item-box/.style={anchor=south,minimum width=0.28cm, font=\footnotesize, draw, inner sep=0cm},
			newest-box/.style={fill=black!10},
			bag-box/.style={anchor=north,font=\footnotesize, inner xsep=0.3mm}
			]
			\matrix[anchor=south, column sep=0mm]
			{
				\draw node[item-box,minimum height=3cm] {12};
				&
				\draw (0,0) node[item-box,minimum height=1.5cm] {6};
				&
				\draw (0,0) node[item-box,minimum height=1.5cm] {6};
				&
				\draw (0,0) node[item-box,newest-box,minimum height=0.75cm] {3};
				\\
				\node[bag-box] {$\xx_1$};
				&
				\node[bag-box] {$\xx_2$};
				&
				\node[bag-box] {$\xx_3$};
				&
				\node[bag-box] {$\xx_4$};
				\\
			};
			\node[font=\footnotesize] {Step 4};
		\end{tikzpicture}
		&
\begin{tikzpicture}
			[
			item-box/.style={anchor=south,minimum width=0.28cm, font=\footnotesize, draw, inner sep=0cm},
			newest-box/.style={fill=black!10},
			bag-box/.style={anchor=north,font=\footnotesize, inner xsep=0.3mm}
			]
			\matrix[anchor=south, column sep=0mm]
			{
				\draw node[item-box,minimum height=3cm] {12};
				&
				\draw (0,0) node[item-box,minimum height=1.5cm] {6};
				&
				\draw (0,0) node[item-box,minimum height=1.5cm] {6};
				&
				\draw (0,0) node[item-box,minimum height=0.75cm] {3};
				\draw (0,0.75) node[item-box,newest-box,minimum height=0.75cm] {3};
				\\
				\node[bag-box] {$\xx_1$};
				&
				\node[bag-box] {$\xx_2$};
				&
				\node[bag-box] {$\xx_3$};
				&
				\node[bag-box] {$\xx_4$};
				\\
			};
			\node[font=\footnotesize] {Step 5};
		\end{tikzpicture}
		&
\begin{tikzpicture}
			[
			item-box/.style={anchor=south,minimum width=0.28cm, font=\footnotesize, draw, inner sep=0cm},
			newest-box/.style={fill=black!10},
			bag-box/.style={anchor=north,font=\footnotesize, inner xsep=0.3mm}
			]
			\matrix[anchor=south, column sep=0mm]
			{
				\draw node[item-box,minimum height=3cm] {12};
				&
				\draw (0,0) node[item-box,minimum height=1.5cm] {6};
				\draw (0,1.5) node[item-box,newest-box,minimum height=0.75cm] {3};
				&
				\draw (0,0) node[item-box,minimum height=1.5cm] {6};
				&
				\draw (0,0) node[item-box,minimum height=0.75cm] {3};
				\draw (0,0.75) node[item-box,minimum height=0.75cm] {3};
				\\
				\node[bag-box] {$\xx_1$};
				&
				\node[bag-box] {$\xx_2$};
				&
				\node[bag-box] {$\xx_3$};
				&
				\node[bag-box] {$\xx_4$};
				\\
			};
			\node[font=\footnotesize] {Step 6};
		\end{tikzpicture}
		&
\begin{tikzpicture}
			[
			item-box/.style={anchor=south,minimum width=0.28cm, font=\footnotesize, draw, inner sep=0cm},
			newest-box/.style={fill=black!10},
			bag-box/.style={anchor=north,font=\footnotesize, inner xsep=0.3mm}
			]
			\matrix[anchor=south, column sep=0mm]
			{
				\draw node[item-box,minimum height=3cm] {12};
				&
				\draw (0,0) node[item-box,minimum height=1.5cm] {6};
				\draw (0,1.5) node[item-box,minimum height=0.75cm] {3};
				&
				\draw (0,0) node[item-box,minimum height=1.5cm] {6};
				\draw (0,1.5) node[item-box,newest-box,minimum height=0.75cm] {3};
				&
				\draw (0,0) node[item-box,minimum height=0.75cm] {3};
				\draw (0,0.75) node[item-box,minimum height=0.75cm] {3};
				\\
				\node[bag-box] {$\xx_1$};
				&
				\node[bag-box] {$\xx_2$};
				&
				\node[bag-box] {$\xx_3$};
				&
				\node[bag-box] {$\xx_4$};
				\\
			};
			\node[font=\footnotesize] {Step 7};
		\end{tikzpicture}
		&
\begin{tikzpicture}
			[
			item-box/.style={anchor=south,minimum width=0.28cm, font=\footnotesize, draw, inner sep=0cm},
			newest-box/.style={fill=black!10},
			bag-box/.style={anchor=north,font=\footnotesize, inner xsep=0.3mm}
			]
			\matrix[anchor=south, column sep=0mm]
			{
				\draw node[item-box,minimum height=3cm] {12};
				&
				\draw (0,0) node[item-box,minimum height=1.5cm] {6};
				\draw (0,1.5) node[item-box,minimum height=0.75cm] {3};
				&
				\draw (0,0) node[item-box,minimum height=1.5cm] {6};
				\draw (0,1.5) node[item-box,minimum height=0.75cm] {3};
				&
				\draw (0,0) node[item-box,minimum height=0.75cm] {3};
				\draw (0,0.75) node[item-box,minimum height=0.75cm] {3};
				\draw (0,1.5) node[item-box,newest-box,minimum height=0.25cm, font=\footnotesize] {1};
				\\
				\node[bag-box] {$\xx_1$};
				&
				\node[bag-box] {$\xx_2$};
				&
				\node[bag-box] {$\xx_3$};
				&
				\node[bag-box] {$\xx_4$};
				\\
			};
			\node[font=\footnotesize] {Step 8};
		\end{tikzpicture}
		&
\begin{tikzpicture}
			[
			item-box/.style={anchor=south,minimum width=0.28cm, font=\footnotesize, draw, inner sep=0cm},
			newest-box/.style={fill=black!10},
			bag-box/.style={anchor=north,font=\footnotesize, inner xsep=0.3mm}
			]
			\matrix[anchor=south, column sep=0mm]
			{
				\draw node[item-box,minimum height=3cm] {12};
				&
				\draw (0,0) node[item-box,minimum height=1.5cm] {6};
				\draw (0,1.5) node[item-box,minimum height=0.75cm] {3};
				&
				\draw (0,0) node[item-box,minimum height=1.5cm] {6};
				\draw (0,1.5) node[item-box,minimum height=0.75cm] {3};
				&
				\draw (0,0) node[item-box,minimum height=0.75cm] {3};
				\draw (0,0.75) node[item-box,minimum height=0.75cm] {3};
				\draw (0,1.5) node[item-box,minimum height=0.25cm, font=\footnotesize] {1};
				\draw (0,1.75) node[item-box,newest-box,minimum height=0.25cm, font=\footnotesize] {1};
				\\
				\node[bag-box] {$\xx_1$};
				&
				\node[bag-box] {$\xx_2$};
				&
				\node[bag-box] {$\xx_3$};
				&
				\node[bag-box] {$\xx_4$};
				\\
			};
			\node[font=\footnotesize] {Step 9};
		\end{tikzpicture}
	\end{tabular}
	\caption{An execution of \Cref{alg:mms-value} to obtain a maximin $4$-partition for the factored utility function $v = (12, 6, 6, 3, 3, 3, 3, 1, 1)$. For each step, the gray box indicates the item placed by the algorithm in that step.}
	\label{fig:mms-value-alg-example}
\end{figure}
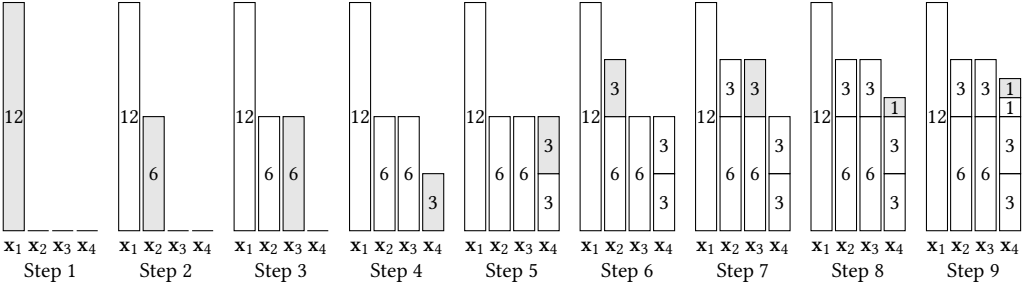

\begin{lemma}\label{lem:mms-partition-algorithm-works}
	Given a factored utility function $v$ over a set of items $\items$ (all goods or all chores), \Cref{alg:mms-value} efficiently computes a maximin $n$-partition of $\items$ under $v$.
\end{lemma}
\begin{proof}
	Let $\xx$ be the partition returned by the algorithm. Without loss of generality, let $r_1,\ldots,r_m$ be the order in which the algorithm considers the items. Then, $v(r_{k+1}) \mid v(r_k)$ for $k \in [m-1]$. Suppose for contradiction that this is not a maximin partition. Among all maximin partitions, choose $\xx'$ such that the lowest index $\ell$ for which $\xx$ and $\xx'$ differ in the assignment of item $r_\ell$ is the maximum possible. 
	
	Let $\yy$ denote the partial allocation of items $r_1,\ldots,r_{\ell-1}$ under both $\xx$ and $\xx'$. Let $i$ and $i'$ be such that $r_\ell \in \xx_i \cap \xx'_{i'}$; note that $i \neq i'$. Because the algorithm assigns $r_\ell$ to bundle $i$ given the partial allocation $\yy$, we must have $|v(\yy_i)| \le |v(\yy_{i'})|$. We consider two cases.
	
	\paragraph{Case 1:} Suppose $|v(\yy_i)| = |v(\yy_{i'})|$. Let $\hat{\xx}$ be a partition obtained by starting from the maximin partition $\xx'$, and swapping the items in $\xx'_i \setminus \yy_i$ and the items in $\xx'_{i'} \setminus \yy_{i'}$ between bundles $i$ and $i'$. Note that $v(\hat{\xx}_i) = v(\xx'_{i'})$ and $v(\hat{\xx}_{i'}) = v(\xx'_i)$. Hence, $\hat{\xx}$ is also a maximin partition. But it matches $\xx$ in the assignment of the first $\ell$ items, which is a contradiction. 
	
	\paragraph{Case 2:} Suppose $|v(\yy_i)| < |v(\yy_{i'})|$. Note that because $v(r_\ell) \mid v(r_k)$ for all $k < \ell$, we must have $|v(\yy_i)| \le |v(\yy_{i'})|-|v(r_\ell)|$. Here, we consider two sub-cases.
	
	\begin{itemize}
		\item \emph{Case 2a:} Suppose $|v(\xx'_i)| < |v(\yy_i)|+|v(r_\ell)| \le |v(\yy_{i'})|$. Let $\hat{\xx}$ be a partition obtained by starting from the maximin partition $\xx'$, and swapping the items in $\xx'_i \setminus \yy_i$ and the item $r_\ell$ between bundles $i$ and $i'$. Note that $|v(\hat{\xx}_i)| = |v(\yy_i)|+|v(r_\ell)| > |v(\xx'_i)|$ and $|v(\hat{\xx}_{i'})| \ge |v(\yy_{i'})| > |v(\xx'_i)|$. Hence, the minimum absolute value across bundles weakly increases when moving from $\xx'$ to $\hat{\xx}$. Similarly, note that $|v(\hat{\xx}_i)| = |v(\yy_i)|+|v(r_\ell)| \le |v(\yy_{i'})| \le |v(\xx'_{i'})|$ and $|v(\hat{\xx}_{i'})| < |v(\yy_{i'})|+|v(r_\ell)| \le |v(\xx'_{i'})|$. Hence, the maximum absolute value across bundles weakly decreases when moving from $\xx'$ to $\hat{\xx}$. This implies that $\hat{\xx}$ is also a maximin partition. However, $\hat{\xx}$ matches $\xx$ in the assignment of the first $\ell$ items, which is a contradiction.
		
		\item \emph{Case 2b:} Suppose $|v(\xx'_i)| \ge |v(\yy_i)|+|v(r_\ell)|$. Suppose $\xx'_i \setminus \yy_i = \set{r_{k_1},r_{k_2},\ldots}$ where $k_1 < k_2 < \ldots$. Let $t$ be the smallest index such that $|v(\yy_i \cup \set{r_{k_1},\ldots,r_{k_t}})| \ge |v(\yy_i)|+|v(r_\ell)|$. Note that $|v(\yy_i \cup \set{r_{k_1},\ldots,r_{k_{t-1}}})| < |v(\yy_i)|+|v(r_\ell)|$. Further, since $v(r_{k_t}) \mid v(r)$ for all $r \in \yy_i \cup \set{r_{k_1},\ldots,r_{k_{t-1}}}$, we must have $|v(\yy_i \cup \set{r_{k_1},\ldots,r_{k_{t-1}}})| \le |v(\yy_i)|+|v(r_\ell)|-|v(r_{k_t})|$. Hence, it must be the case that $|v(\yy_i \cup \set{r_{k_1},\ldots,r_{k_t}})| = |v(\yy_i)|+|v(r_\ell)|$, i.e., $|v(\set{r_{k_1},\ldots,r_{k_t}})| = |v(r_\ell)|$. In this case, swapping the set of items $\set{r_{k_1},\ldots,r_{k_t}}$ with the item $r_\ell$ between bundles $i$ and $i'$ in $\hat{\xx}$ produces another maximin partition which matches $\xx$ in the assignment of the first $\ell$ items, which is a contradiction. 
	\end{itemize} 
	This completes the proof.
\end{proof}
 
When we assume our instance to be ordered, we will consider the items in the standard order $1,\ldots,m$ in \Cref{alg:mms-value}. This will allow us to reason about the exact indices of items allocated to different bundles under \Cref{alg:mms-value}. 

\Cref{alg:mms-value} does not always work for utility functions $v$ that are not factored. Consider, for example, $v = (3, 3, 2, 2, 2)$ and $n = 2$. The algorithm produces the partition $\xx_1 = \{3, 2, 2\}$ and $\xx_2 = \{3, 2\}$, achieving a minimum share of $5$. However, the unique MMS partition is given by $\xx_1 = \{3, 3\}$ and $\xx_2 = \{2, 2, 2\}$, which achieves a maximin share of $6$.

\subsection{Weakly Lexicographic Utilities}

We now present a valid reduction for weakly lexicographic utilities. First, we introduce the concept of a ``bad cut''. Recall that we work with ordered instances in which $|v_i(1)| \ge \ldots \ge |v_i(m)|$ for all $i$. 

\begin{definition}[Bad Cuts]
	In a goods or chore division instance $I = (\agents,\items,\vv)$, we say that index $k \in [m-1]$ is a \emph{cut} of agent $i$ if $v_i(k) \neq v_i(k+1)$. Further, if $k$ is not a multiple of $n$, we say that it is a \emph{bad cut} of agent $i$. Define $C_i$ to be the smallest bad cut of agent $i$; let $C_i = m$ if agent $i$ does not have any bad cuts.
\end{definition}

\begin{table}[th]
	\centering
\setlength\arrayrulewidth{0.1pt}
	\begin{tabular}{l |  l l l  |  l l l |  l l l | l}
		\toprule
		$i$ & $\abs{v_1}$ & $\abs{v_2}$ & $\abs{v_3}$ & $\abs{v_4}$ & $\abs{v_5}$ & $\abs{v_6}$ & $\abs{v_7}$ & $\abs{v_8}$ & $\abs{v_9}$ & $C_i$ \\
		\midrule
		$i_1$ & 81 & 81 & 81 & 81${|_{4}}$& 9 & 9 & 9${|_{7}}$ & 1 & 1 & 4 \\
		$i_2$ & 81 & 81 & 81$|_\checkmark$ & 9 & 9 & 9$|_\checkmark$ & 1 & 1  & 1 & 9 \\
		$i_3$ & 729${|_{1}}$ & 81 & 81 & 81${|_{4}}$ & 9 & 9 & 9${|_{7}}$ & 1 & 1 & 1 \\
\bottomrule
	\end{tabular}
	\caption{An instance with weakly lexicographic utilities. Bad cuts at index $k$ are shown as $|_k$, and non-bad cuts as $|_\checkmark$.}
	\label{table:wolex-example}
	\vspace{-10pt}
\end{table}

\begin{example}
	\label{example:wolex-example}
	The ordered instance described in \Cref{table:wolex-example} consists of $n = 3$ agents and $ m =9$ items. The \emph{cuts} of $i_1$, $i_2$, and $i_3$ are $\{4, 7\}$, $\{3, 6\}$, and $\{1, 4, 7\}$ respectively. Here, a cut is considered a \emph{bad cut} if it is not divisible by $n = 3$. Then, all cuts of $i_1$ and $i_3$ are bad while $i_2$ has no bad cuts. Following the definition of $C_i$'s, we have $C_{i_1} = 4$, $C_{i_2} = m = 9$, and $C_{i_3} = 1$.
\end{example}

First, for any agent $i$, we identify a specific bundle in a maximin $n$-partition of agent $i$ produced by \Cref{alg:mms-value}, in terms of $C_i$. 

\begin{lemma}
	\label{lem:wolex-one-partition}
	For a goods or chore division instance $I = (\agents, \items, \vv)$ with weakly lexicographic utilities and agent $i \in \agents$, there exists a maximin $n$-partition of agent $i$ in which one of the bundles is $S = \set{1, n + 1, \ldots, kn + 1}$, where $k = \floor{(C_i-1)/n}$.
\end{lemma}
\begin{proof}
	Note that because $C_i$ is the smallest bad cut of agent $i$, we have that for each $k' \in [k]$, agent $i$ has equal utility for all items in $[(k' - 1)n + 1, k'n]$. Consider how \Cref{alg:mms-value} constructs a maximin $n$-partition for agent $i$ given $v_i$. First, for each $k' \in [k]$, it divides items in $[(k'-1)n+1,k'n]$ equally between the bundles (one to each). Then, it assigns items $[kn + 1, C_i]$ to $C_i - kn$ many bundles, again one to each. Note that $C_i-kn \ge 1$ by the definition of $k$. Hence, the first bundle is precisely $S$. 
	
	Note that these $C_i-kn$ bundles receive an extra item $[kn + 1, C_i]$ from compared to the remaining $n-(C_i-kn)$ bundles. Further, because $v_i$ is weakly lexicographic and $C_i$ is a bad cut, this item has more absolute value than all items indexed greater than $C_i$ combined. Thus, \Cref{alg:mms-value} divides the remaining items $[C_i+1,m]$ among the remaining $n - (C_i - kn)$ bundles. In particular, bundle $1$ does not receive any of these items. Hence, in the end, bundle $1$ contains exactly the set of items $S$, as needed. 
\end{proof}

\Cref{lem:wolex-one-partition} shows that in \Cref{example:wolex-example} there exists a maximin $3$ partition for $i_1$ where one of the bundles is $\{1, 4\}$. Same applies to  $\{1, 4, 7\}$ and $\{1\}$ for $i_1$, $i_2$ and $i_3$ respectively. Following this observation and  assuming items here are all goods, allocating $\{1\}$ to $i_3$ is a valid reduction by \Cref{lem:bag-reduction}, because $\{1\}$ is a subset of the other two bundles described. If the items were chores, i.e.\ values being costs, allocating $\{1, 4, 7\}$ to $i_2$ would have formed a valid reduction.

Next, we show that if we choose an agent $i$ with the minimum or maximum $C_i$ and the corresponding $S$ from \Cref{lem:wolex-one-partition}, the pair $(i,S)$ forms a valid reduction. Note that this valid reduction can be found in polynomial time. 

\begin{lemma}
	\label{lem:wolex-reduction}
	For a goods (respectively, chore) division instance $I = (\agents, \items, \vv)$ with weakly lexicographic utilities, the pair $(i, S)$ is a valid reduction  when $i$ is an agent with the minimum (resp., maximum) $C_i$ and $S = \set{1, n + 1, \ldots, kn + 1}$, where $k = \floor{(C_i-1)/n}$. 
\end{lemma}
\begin{proof}
	By \Cref{lem:wolex-one-partition}, we have $S$ is one of the bundles in a maximin $n$-partition of agent $i$; hence, $v_i(S) \ge \MMS_i$. 
	
	Consider any other agent $i'$, define $k' = \floor{(C_{i'}-1)/n}$, and let $S' = \set{1,n+1,\ldots,k'n+1}$. Then, by \Cref{lem:wolex-one-partition}, $S'$ is one of the bundles in a maximin $n$-partition of agent $i'$. Further, by our choice of agent $i$, we have $k \le k'$ (thus, $S \subseteq S'$) for goods division, and $k \ge k'$ (thus, $S \supseteq S'$) for chore division. It is easy to see that $S$ and $S'$ satisfy the necessary conditions from \Cref{lem:bag-reduction}. Hence, $(i,S)$ is a valid reduction. \end{proof}

\subsection{Factored Personalized Bivalued Utilities}

In this section, we present a valid reduction for factored personalized bivalued utilities. Recall that we work with ordered instances. Hence, for each agent $i$, there exists $k \in [m]$ such that $|v_i(r)| = p_i$ for all $r \le k$ and $|v_i(r)| = 1$ for all $r > k$. Thus, each agent $i$ has at most one cut ($k$, if $k < m$), and $C_i$ is equal to this cut (if it exists and it is bad) and $m$ otherwise. However, in this case, simply choosing an agent $i$ with the minimum or maximum $C_i$ does not work. Instead, we rely on a different metric, called ``idle time''.

\begin{definition}[Idle Time]
	In a goods or chore division instance $I = (\agents,\items,\vv)$ with factored personalized bivalued utilities, we define $AC_i$ as $0$ if $C_i = m$ and as $n - (C_i \bmod n)$ otherwise.
	Let the \emph{idle time} of agent $i$ to be
	$
	T^{\idle}_i = \min\set{p_i \cdot AC_i, m - C_i}.
	$
\end{definition}
	
\begin{table}[h]
	\centering
\setlength\arrayrulewidth{0.1pt}
	\begin{tabular}{l  l l l  |  l l l  |   l l l   |  l l l}
		\toprule
		$i$ & $\abs{v_1}$ & $\abs{v_2}$ & $\abs{v_3}$ & $\abs{v_4}$ & $\abs{v_5}$ & $\abs{v_6}$ & $\abs{v_7}$ & $\abs{v_8}$ & $\abs{v_9}$ & $C_i$ & $AC_i$ & $T^{\idle}_i$
		\\
		\midrule
		$i_1$ & 2 & 2 & 2 & 2 $|_4$ & \underline{\textbf{1}} & \underline{\textbf{1}}& \underline{1} & \underline{1}  & 1 &  4 & 2 & 4
		\\ 
		$i_2$ & 5 $|_1$ & \underline{\textbf{1}}& \underline{\textbf{1}}& \underline{1} & \underline{1} & \underline{1} & \underline{1} & \underline{1} & \underline{1} & 1 & 2 & 8
		\\ 
		$i_3$ & 4 & 4 & 4 & 4 & 4 & 4 & 4 & 4 $|_8$ &  \underline{\textbf{1}} & 8 & 1 & 1
		\\ 
		\bottomrule
	\end{tabular}
	\caption{An instance with personalized bivalued utilities. Bold cells refer to active bundles $AC_i$, and the underlined ones refer to idle times $T^{\idle}_i$.}
	\label{table:personalized-bival-example}
	\vspace{-15pt}
\end{table}

\begin{example}
	\label{example:personalized-bival-example}
	\Cref{table:personalized-bival-example} presents an instance with personalized bivalued utilities where $p_{i_1} = 2$, $p_{i_2} = 5$ and $p_{i_3} = 4$.  The number of \emph{active bundles} and \emph{idle times} are also shown in the table.
\end{example}

First, note that when agent $i$ does not admit a bad cut, we have $C_i = m$, $AC_i = 0$, and $T^\idle_i = 0$. Suppose agent $i$ admits a bad cut $C_i = kn + r$ with remainder $r \in [n-1]$. Observe that \Cref{alg:mms-value} operates in at most three phases. In the first phase, it divides the items with absolute value $p_i$ in a round robin fashion between all $n$ bundles, until it reaches the bad cut. At that time, $C_i \bmod n$ bundles have an extra item with absolute value $p_i$. We refer to the remaining bundles as the \emph{``active bundles''}; note that there are precisely $AC_i$ many active bundles. In the second phase, it divides the items with absolute value $1$ between the active bundles in a round robin fashion, until either all items are allocated or all $n$ bundles become of exactly equal value (this is where the assumption of the utilities being factored, i.e., $p_i$ being an integer is crucial). Note that the duration of this second phase is precisely the idle time of agent $i$ defined above. If there are any remaining items with absolute value $1$, the algorithm divides them between all $n$ bundles in a round robin fashion in the final phase. 

Using this observation, we are ready to characterize one of the bundles in some maximin $n$-partition of agent $i$. 

\begin{lemma}
	\label{lem:factored-bival-one-partition}
	For a goods or chore division instance $I = (\agents, \items, \vv)$ with factored personalized bivalued utilities and agent $i \in \agents$, there exists a maximin $n$-partition of agent $i$ in which one of the bundles is $S = \set{1, n + 1, \ldots, kn + 1}$, where $k = \floor{(m - \max\set{T^\idle_i - AC_i, 0}-1)/n}$.
\end{lemma}
\begin{proof}
	Fix an agent $i \in \agents$. First, if agent $i$ does not have a bad cut, then $T^\idle_i = AC_i = 0$, so $m - \max\set{(T^\idle_i - AC_i), 0} = m$. In this case, \Cref{alg:mms-value} simply allocates all items between $n$ bundles in a round robin fashion, so $S$ coincides with the first bundle. 
	
	Next, suppose agent $i$ has a bad cut $C_i = k'n+r$, where $r \in [n-1]$. In this case, the algorithm divides the first $C_i$ items of absolute value $p_i$ between all $n$ bundles in a round robin fashion, after which point the first bundle has items $\set{1,n+1,\ldots,k'n+1}$. After this, \Cref{alg:mms-value} enters the second phase of allocating items of absolute value $1$ between the active bundles in a round robin fashion, which runs for $T^\idle_i$ steps. 
	
	If $T^\idle_i \le AC_i$, then there must be at most $AC_i$ items left after the bad cut $C_i$. Thus, $k'n +1 \le C_i \le m \le C_i + AC_i = (k' + 1)n$. This implies that $k = \floor{(m-1)/n} = k'$. Hence, $S$ is precisely the first bundle produced by \Cref{alg:mms-value}.
	
	Finally, suppose $T^\idle_i \ge AC_i$. Since all items in $[C_i + 1, m]$ have absolute value $1$, we can do the following during the second phase of \Cref{alg:mms-value}: first allocate items $\set{C_i+1,\ldots,C_i+AC_i}$ to the active bundles, one each, and then allocate items $\set{m -(T^\idle_i - AC_i)+1,\ldots,m}$ to the active bundles in a round robin fashion. Items $\set{C_i+AC_i+1,\ldots,m -(T^\idle_i - AC_i)}$ (if any) are reserved for the third phase in which items with absolute value $1$ need to be divided in a round robin fashion between all $n$ bundles. This change can be interpreted as running \Cref{alg:mms-value} with a different tie-breaking among items with absolute value $1$. Hence, by \Cref{lem:mms-partition-algorithm-works}, this still produces a maximin partition. Under this partition, the items allocated to bundle $1$ (which is necessarily \emph{not} an active bundle) are those that would be allocated if we allocate items $\set{1,\ldots,m -(T^\idle_i - AC_i)}$ in a round robin fashion, i.e., $S = \set{1,n+1,\ldots,kn+1}$, where $k = \floor{(m -(T^\idle_i - AC_i)-1)/n}$, as needed. 
\end{proof} 
In \Cref{example:personalized-bival-example}, we can find a valid reduction similar to the case of weakly lexicographic utilities. By \Cref{lem:factored-bival-one-partition}, there is a maximin $3$ partition for $i_1$ where one of the bundles is $\{1, 4, 7\}$. Same holds for bundles $\{1\}$ and $\{1, 4, 7\}$ for $i_2$ and $i_3$ respectively. If items were goods, the pair of  $i_2$ with $\{1\}$ would have been a valid reduction, and the pair $i_1$ and $\{1, 4, 7\}$ would have worked if they were chores.

Now, we can show that choosing agent $i$ with the minimum or maximum $\max\set{(T^\idle_i - AC_i), 0}$ and the corresponding $S$ from \Cref{lem:factored-bival-one-partition} yields a valid reduction $(i,S)$. 

\begin{lemma}
	\label{lem:valid-reduction-factored-bival}
	For a goods (respectively, chore) division instance $I = (\agents, \items, \vv)$ with weakly lexicographic utilities, the pair $(i, S)$ is a valid reduction when $i$ is an agent with the maximum (resp., minimum) value of $\max\set{(T^\idle_i - AC_i), 0}$ and $S = \set{1, n + 1, \ldots, kn + 1}$, where $k = \floor{(m-\max\set{(T^\idle_i - AC_i), 0}-1)/n}$. 
\end{lemma}
\begin{proof}
	By \Cref{lem:factored-bival-one-partition}, we know that $S$ is one of the bundles in a maximin $n$-partition of agent $i$; hence, $v_i(S) \ge \MMS_i$. 
	
	For any other agent $i'$, define $S' = \set{1, n + 1, \ldots, k'n + 1}$, where $k' = \floor{(m-\max\set{(T^\idle_{i'} - AC_{i'}), 0}-1)/n}$. Then, by \Cref{lem:factored-bival-one-partition}, $S'$ is one of the bundles in a maximin $n$-partition of agent $i'$. Further, due to our choice of $i$ and using the same reasoning as used in the proof of \Cref{lem:wolex-reduction}, we have that $S \subseteq S'$ for goods division and $S \supseteq S'$ for chore division, which satisfies the condition of \Cref{lem:bag-reduction}.
\end{proof} 
\subsection{Achieving Pareto Optimal MMS Allocations}\label{sec:mms-po}

In this section, we show that for weakly lexicographic as well as for factored bivalued instances, we can compute an allocation that is MMS and PO in polynomial time. Our approach uses the fact that if $\xx$ is an MMS allocation, and $\xx'$ is a Pareto improvement over $\xx$ then $\xx'$ is also MMS. Thus to find an MMS and PO allocation, we can compute an MMS allocation using \Cref{thm:mms} and then repeatedly find Pareto improvements until we reach a PO allocation. In this section, we will show that we can in polynomial time find Pareto improvements if they exist, and that we will reach a PO allocation after at most polynomially many Pareto improvements.

\citet{aziz2019efficient} prove that in case of goods division with weakly lexicographic or bivalued utilities, one can efficiently test if a given allocation is Pareto optimal (PO). Further, if it is not PO, a Pareto dominating allocation with special properties always exists and can be computed efficiently. The following lemma states their result for goods division, together with an extension to chore division. While in the case of weakly lexicographic utilities our proof for chores almost mirrors their proof for goods, the ideas needed in the case of bivalued utilities are slightly different for chores. Also, the statement below is their claim for weakly lexicographic utilities; while they make a differently worded claim for bivalued utilities, their proof also shows that this claim holds for bivalued utilities. 

\begin{lemma}\label{lem:pareto-improve}
	In a goods or chore division instance with weakly lexicographic or bivalued utilities, one can efficiently test whether a given allocation $\xx$ is Pareto optimal. Further, if $\xx$ is not Pareto optimal, then there exists a cycle of distinct agents $(i_1,\ldots,i_k,i_{k+1}=i_1)$ and a cycle of distinct items $(r_1,\ldots,r_k,r_{k+1}=r_1)$ such that:
	\begin{enumerate}
		\item $r_t \in \xx_{i_t}$ and $v_{i_t}(r_{t-1}) \ge v_{i_t}(r_t)$ for each $t \in \set{2,\ldots,k+1}$, 
		\item at least one of the above inequalities is strict, and 
		\item the allocation $\xx^*$ obtained from $\xx$ by reallocating item $r_{t-1}$ to agent $i_t$ for each $t \in \set{2,\ldots,k+1}$ is a Pareto improvement over $\xx$. 
	\end{enumerate}
	Such a Pareto improvement $\xx^*$ can be computed in polynomial time. 
\end{lemma}
Given bivalued utilities with values $0 < a < b$ (goods division) or $0 > a > b$ (chore division), an allocation $\yy$, and an agent $i$, define $\yy^+_i$ (resp., $\yy^-_i$) as the set of items in $\yy_i$ for which agent $i$ has value $b$ (resp., $a$). 

\begin{proof}
	The goods division case is proved by \citet{aziz2019efficient}. Let us focus on chore division. First, let us establish the existence of the special Pareto improvement $\xx^*$ in case $\xx$ is not PO. Note that if we establish the existence of the desired cycles of agents and items, then the first two properties of these cycles claimed in the lemma imply that the reallocation that yields $\xx^*$ makes each agent weakly better off and some agent strictly better off, i.e., that it is a Pareto improvement. 
	
	\paragraph{Chore division, weakly lexicographic utilities:} Let $I=(\agents,\items,\vv)$ be a chore division instance with weakly lexicographic utilities and $\xx$ be an allocation that is not Pareto optimal. Among all Pareto improvements, let $\hat{\xx}$ be the one that is closest to $\xx$ in that it minimizes $|\bigcup_{i \in \agents} \hat{\xx}_i \setminus \xx_i|$. 
	
	Consider an agent $i_1$ who is strictly better off under $\hat{\xx}$ than under $\xx$; such an agent must exist in a Pareto improvement. Then, there must exist a chore $c_1 \in \xx_{i_1} \setminus \hat{\xx}_{i_1}$ which $i_1$ has given away under $\hat{\xx}$. Let $i_2 \neq i_1$ be the agent who received $c_1$, so that $c_1 \in \hat{\xx}_{i_2}$. Because agent $i_2$ is weakly better off under $\hat{\xx}$ than under $\xx$ but has received chore $c_1$ with $v(c_1) < 0$, she must have given away at least one chore. Because her utility function is weakly lexicographic, in fact she must have lost a chore $c_2 \in \xx_{i_2} \setminus \hat{\xx}_{i_2}$ with $v_{i_2}(c_1) \ge v_{i_2}(c_2)$. 
	
	More generally, for $t \ge 2$, suppose we obtain a sequence of chores $(c_1,\ldots,c_{t-1})$ and a sequence of agents $(i_1,\ldots,i_t)$ such that chore $c_k$ is transferred from agent $i_k$ to agent $i_{k+1}$ for each $k \in [t-1]$. Then, because agent $i_t$ receives chore $c_{t-1}$ under $\hat{\xx}$, and she is weakly better off under $\hat{\xx}$ than under $\xx$, and her utility function is weakly lexicographic, she must have lost a chore $c_t \in \xx_{i_t} \setminus \hat{\xx}_{i_t}$ with $v_{i_t}(c_{t-1}) \ge v_{i_t}(c_t)$ to another agent $i_{t+1}$, and the sequence continues. Since the number of agents is finite, this process must run into a cycle where $i_{t+1} = i_{\ell}$ for some $\ell < t$. We can picture this situation as follows.
	
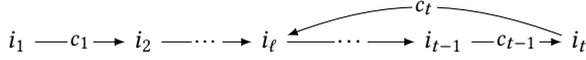
\[
	\begin{tikzpicture}
		\node (i1) {$i_1$};
		\node (i2) [right=1.2 of i1] {$i_2$};
		\node (i3) [right=1.2 of i2] {$i_\ell$};
		\node (it1) [right=1.7 of i3] {$i_{t-1}$};
		\node (it) [right=1.2 of it1] {$i_{t}$};
		\draw[-latex] 
			(i1) -- (i2) 
			node [midway, font=\small, inner sep=1pt, fill=white] {$c_1$};
		\draw[-latex] 
			(i2) -- (i3) 
			node [midway, font=\small, inner sep=1pt, fill=white] {\dots};
		\draw[-latex] (i3) -- (it1) node [midway,inner sep=1pt,fill=white] {\dots};
	
		\draw[-latex] 
			(it1) -- (it) 
			node [midway, font=\small, inner sep=1pt, fill=white] {$c_{t-1}$};
		\draw[-latex] 
			(it) to [bend right=20] 
			node [midway, font=\small, inner sep=1pt, fill=white] {$c_t$}
			(i3);
	\end{tikzpicture}
	\]
	
We now consider the cycle $(i_\ell, \dots, i_t, i_\ell)$. If we have $v_{i_{\ell}}(c_t) > v_{i_{\ell}}(c_{\ell})$ or $v_{i_k}(c_{k-1}) > v_{i_k}(c_k)$ for some $k \in \set{\ell+1,\ldots,t}$, then the cycle satisfies the conditions of the lemma and we are done. Otherwise we have $v_{i_{\ell}}(c_t) = v_{i_{\ell}}(c_{\ell})$ and $v_{i_k}(c_{k-1}) = v_{i_k}(c_k)$ for all $k \in \set{\ell+1,\ldots,t}$. In this case, the allocation obtained by starting from $\hat{\xx}$ and reassigning chore $c_k$ back to agent $i_k$ for each $k \in \set{\ell,\ldots,t}$ is still a Pareto improvement over $\xx$ (because all agents are indifferent between this new allocation and $\hat{\xx}$) and is closer to $\xx$, which contradicts our choice of $\hat{\xx}$.
	
	\paragraph{Chore division, bivalued utilities:} Let $I = (\agents,\items,\vv)$ be a chore division instance with bivalued utilities and $\xx$ be an allocation that is not Pareto optimal. Among all Pareto improvements, choose $\hat{\xx}$ to be the closest to $\xx$ in the sense of minimizing $|\bigcup_{i \in \agents} \hat{\xx}_i \setminus \xx_i|$. 
	
	First, we show that there is no \emph{clear winner} $i$ in $\hat{\xx}$ for whom $\hat{\xx}_i \subsetneq \xx_i$. If this were the case, we could take a chore $c \in \xx_i \setminus \hat{\xx}_i$ and give it back to agent $i$ in $\hat{\xx}$. The resulting allocation would still be a Pareto improvement over $\xx$ (agent $i$ is still weakly better off, and the agent who gives $c$ back must now be strictly better off), and it would be closer to $\xx$, which contradicts our choice of $\hat{\xx}$. 
	
	Next, we show that $|\bigcup_{i \in \agents} \hat{\xx}^+_i| < |\bigcup_{i \in \agents} \xx^+_i|$, i.e., $\hat{\xx}$ allocates strictly fewer chores to agents who find them difficult than does $\xx$. Note that for any allocation $\yy$, the social welfare (the sum of utilities of agents) under $\yy$ is 
	\[
	\textstyle
	|\bigcup_{i \in \agents} \yy^+_i| \cdot b + |\bigcup_{i \in \agents} \yy^-_i| \cdot a = |\bigcup_{i \in \agents} \yy^+_i| \cdot (b-a) + m \cdot a.
	\]
	Because $\hat{\xx}$ is a Pareto improvement over $\xx$, the social welfare under $\hat{\xx}$ is strictly higher than the social welfare under $\xx$. Since $b-a < 0$, it follows that $|\bigcup_{i \in \agents} \hat{\xx}^+_i| < |\bigcup_{i \in \agents} \xx^+_i|$.
	
	Consider a chore $c_1 \in \bigcup_{i \in \agents} \xx^+_i \setminus \bigcup_{i \in \agents} \hat{\xx}^+_i $. Suppose $c_1 \in \xx^+_{i_1} \cap \hat{\xx}^-_{i_2}$ for some $i_2 \neq i_1$. We can represent this as $i_1 \stackrel[\xx]{b}{\to} c_1 \stackrel[\hat{\xx}]{a}{\to} i_2$, where, for an arrow connecting agent $i$ with chore $c$, the entry above indicates $v_i(c)$ while the entry below indicates the allocation in which $c$ is allocated to $i$. Consider extending this chain as much as possible by adding alternating $\stackrel[\xx]{a}{\to}$ and $\stackrel[\hat{\xx}]{a}{\to}$ edges to obtain $i_1 \stackrel[\xx]{b}{\to} c_1 \stackrel[\hat{\xx}]{a}{\to} i_2 \ldots c_{t-1} \stackrel[\hat{\xx}]{a}{\to} i_t$. There are two possibilities: either the chain stops at agent $i_t$ for some $t \ge 2$ (and we are unable to extend it further), or an agent repeats at some point (i.e., $i_t = i_\ell$ for some $\ell < t$). 
	
	\paragraph{Case 1: the chain stops at agent $i_t$.} First, suppose $\xx^+_{i_t} \neq \emptyset$. Consider any chore $\hat{c}_t \in \xx^+_{i_t}$. Consider the allocation obtained by starting from $\xx$ and cyclically shifting chores as follows: chore $c_k$ is moved to agent $i_{k+1}$ for $k \in [t-1]$, and chore $\hat{c}_t$ is moved to agent $i_1$. Note that agent $i_1$ loses a $b$-valued chore and gains a chore, agents $i_2$ through $i_{t-1}$ each lose an $a$-valued chore and gain an $a$-valued chore, and agent $i_t$ loses a $b$-valued chore and gains an $a$-valued chore. Thus, this is the kind of cycle sought in the lemma. 

	Next, suppose $\xx^+_{i_t} = \emptyset$. Because $\hat{\xx}$ is a Pareto improvement over $\xx$ in which agent $i_t$ gains a new chore $c_{t-1}$, she must have also lost at least one chore. Pick $c_t \in \xx^-_{i_t} \setminus \hat{\xx}_{i_t}$. Let $c_t \in \hat{\xx}_{i_{t+1}}$. If $c_t \in \hat{\xx}^-_{i_{t+1}}$, then the chain could have continued. Hence, it must be the case that $c_t \in \hat{\xx}^+_{i_{t+1}}$. In this case, consider the allocation obtained by starting from $\hat{\xx}$ and exchanging chores $c_{t-1}$ and $c_t$ between agents $i_t$ and $i_{t+1}$. Note that the utility to agent $i_t$ does not change because she loses an $a$-valued chore and gains an $a$-valued chore, and agent $i_{t+1}$ is weakly better because she loses a $b$-valued chore and gains a chore. Hence, the resulting allocation is still a Pareto improvement over $\xx$. However, it is also closer to $\xx$ than $\hat{\xx}$ is, because we give chore $c_t$ back to agent $i_t$ during the exchange. This contradicts the definition of $\hat{\xx}$. 
	
	\paragraph{Case 2: $i_t = i_\ell$ for some $\ell < t$.} First, suppose $\ell=1$. Then, consider the allocation obtained by starting from $\xx$ and cyclically shifting chores as follows: chore $c_k$ is moved to agent $i_{k+1}$ for $k \in [t-1]$. Note that agent $i_1$ loses a $b$-valued chore and gains an $a$-valued chore, and agents $i_2$ through $i_{t-1}$ each lose an $a$-valued chore and gain an $a$-valued chore. Thus, this is the kind of cycle sought in the lemma. 

	Finally, suppose $\ell \neq 1$. In this case, consider the allocation obtained by starting from $\hat{\xx}$ and cyclically shifting the chores back as follows: chore $c_k$ is moved back to agent $i_k$ for $k \in \set{\ell,\ell+1,\ldots,t-1}$. Compared to $\hat{\xx}$, agents $i_{\ell}$ through $i_{t-1}$ each lose an $a$-valued chore and gain an $a$-valued chore. Hence, the resulting allocation is still a Pareto improvement over $\xx$, and it is closer to $\xx$ than $\hat{\xx}$ is, which contradicts the definition of $\hat{\xx}$. 
	
	\paragraph{Efficient computation:} Finally, for finding the kind of cycles sought in the lemma, we can use the same method that \citet{aziz2019efficient} use for weakly lexicographic utilities. We can create a directed graph with the items as the nodes, and add an edge $(r,r')$ whenever $v_i(r') \ge v_i(r)$ for the agent $i$ who holds item $r$ under $\xx$. We call this edge strict if $v_i(r') > v_i(r)$. Then, the problem reduces to testing the existence of a cycle in this graph with at least one strict edge (and finding it if it exists). This can be done efficiently by considering each strict edge $(r,r')$, and trying to find a path in the graph from $r'$ to $r$. If a cycle is found, the desired Pareto improvement $\xx^*$ can be computed efficiently by reallocating items along the cycle.
\end{proof}
 
For bivalued instances, the conclusion of \Cref{lem:pareto-improve} also follows from our \Cref{thm:fpo}, proved in \Cref{sec:fpo}, which shows that for bivalued utilities Pareto optimality and fractional Pareto optimality are equivalent, together with the fact that fractional Pareto optimality can be checked in polynomial time via linear programming. In fact, the proofs of \Cref{lem:pareto-improve} and \Cref{thm:fpo} are very similar.

The next lemma shows that starting from any allocation, if we repeatedly find a Pareto improvement by invoking \Cref{lem:pareto-improve}, then we arrive at a Pareto optimal allocation in at most a polynomial number of steps. 

\begin{lemma}\label{lem:pareto-chain}
Let $\xx^0$ be an allocation in a goods division or chore division instance with weakly lexicographic or bivalued utilities. Let $(\xx^0,\xx^1,\xx^2,\ldots)$ be a chain in which, for each $k \ge 1$, $\xx^k$ is a Pareto improvement over $\xx^{k-1}$ satisfying the properties in \Cref{lem:pareto-improve}. Then, the chain terminates at a Pareto optimal allocation in at most a polynomial number of steps. 
\end{lemma}
\begin{proof}
	For bivalued utilities, note that the Pareto improvement $\xx^*$ identified in \Cref{lem:pareto-improve} strictly increases (resp., reduces) the number of goods (resp., chores) allocated to agents who value it at $b$. Since this value is between $0$ and $m$, the chain must end in at most $m$ steps.
	
	Next, consider an instance $I = (\agents,\items,\vv)$ with weakly lexicographic utilities. Let us define a quantity $h(i,r)$ for every agent $i$ and item $r$: if $(L_1,\ldots,L_k)$ is the partition of $\items$ under the weakly lexicographic utility function $v_i$ of agent $i$ as in \Cref{def:wolex} and $r \in L_t$, then we set $h(i,r) = t$. For an allocation $\yy$, define the potential function $\phi(\yy) = \sum_{i \in \agents} \sum_{r \in \yy_i} h(i,r)$. Note that $m \le \phi(\yy) \le m^2$. We show that in case of goods division (resp., chore division), every Pareto improvement in the chain strictly decreases (resp., increases) the potential. This implies that the chain must terminate in $O(m^2)$ steps. 

	Consider any Pareto improvement from $\xx$ to $\xx^*$ obtained by a cycle of items $(r_1,\ldots,r_k,r_{k+1}=r_1)$ as in \Cref{lem:pareto-improve}. For any agent $i$, note that for every item $r_t$ that she loses, she gains a unique item $r_{t-1}$ with $v_i(r_{t-1}) \ge v_i(r_t)$, which implies $h(i,r_{t-1}) \le h(i,r_t)$ for goods division and the opposite inequality for chore division. Thus, $\sum_{r \in \xx^*_i} h(i,r) \le \sum_{r \in \xx_i} h(i,r)$ for goods division and the opposite inequality holds for chore division. Further, because the Pareto improvement strictly improves the utility of some agent, the inequality for that agent is strict. Hence, the Pareto improvement strictly decreases (resp., increases) the potential value for goods division (resp., chore division), as desired. 
\end{proof} 
In \Cref{lem:pareto-chain}, note that if the initial allocation $\xx$ is an MMS allocation, then the final allocation must be both MMS and PO, since Pareto improvements preserve the MMS property. Plugging in the MMS allocation obtained in \Cref{thm:mms} as the initial allocation, we obtain the following result. 

\begin{corollary}\label{cor:mms-po}
	In every goods division or chore division instance with weakly lexicographic or factored bivalued utilities, an MMS and PO allocation always exists and can be computed in polynomial time. 
\end{corollary}

Note that none of our results about PO in this section apply to \emph{personalized} bivalued utilities. Obtaining a polynomial-time algorithm for finding an MMS and PO allocation for personalized bivalued instances remains an open problem.

 \section{Discussion}\label{sec:discussion}

We make progress on the open question regarding the existence of an envy-free up to one item (EF1) and Pareto optimal (PO) allocation of chores, by giving a positive answer for the special case when agents have bivalued utilities (i.e., all utilities are in $\set{a,b}$ for some $0 > a > b$). Our algorithm uses the Fisher market framework, which has been used successfully for allocating goods~\cite{barman2018finding}, but requires novel ideas to adapt it to allocate chores. In case of goods with bivalued utilities, \citet{amanatidis2021maximum} show that an allocation satisfying the stronger fairness guarantee of envy-freeness up to any good (EFX) always exists and can be computed efficiently; they also establish the existence of an EFX + PO allocation. \citet{garg2021computing} improve upon this by using the Fisher market framework to compute an EFX + PO allocation efficiently. Investigating whether EFX or EFX + PO allocations of \emph{chores} always exist with bivalued utilities and, if so, whether they can be computed efficiently is an exciting future direction. Alternatively, establishing the existence (and efficient computation) of EF1 + PO allocations of chores under other natural classes of utility functions, such as weakly lexicographic utilities, is also an appealing avenue for future work. Yet another direction would be to adapt our algorithm to achieve EF1 + PO allocations of \emph{mixed items} (where some are goods but others are chores), at least under restricted utilities. 

Regarding our results on maximin share fairness (MMS), recall that the existence of an MMS allocation immediately implies the existence of an MMS + PO allocation because Pareto improvements preserve the MMS guarantee. However, computing an MMS + PO allocation may not always be easy, even when computing an MMS allocation is. To the best of our knowledge, our result is the first to establish non-trivial efficient computation of an MMS + PO allocation under a natural class of utility functions. It would be interesting to try to achieve MMS for more general classes of utility functions, such as general bivalued utilities (when $b/a$ is not an integer) or all factored valuations.

\bibliographystyle{ACM-Reference-Format}

\begin{thebibliography}{37}
	
	
	\ifx \showCODEN    \undefined \def \showCODEN     #1{\unskip}     \fi
	\ifx \showDOI      \undefined \def \showDOI       #1{#1}\fi
	\ifx \showISBNx    \undefined \def \showISBNx     #1{\unskip}     \fi
	\ifx \showISBNxiii \undefined \def \showISBNxiii  #1{\unskip}     \fi
	\ifx \showISSN     \undefined \def \showISSN      #1{\unskip}     \fi
	\ifx \showLCCN     \undefined \def \showLCCN      #1{\unskip}     \fi
	\ifx \shownote     \undefined \def \shownote      #1{#1}          \fi
	\ifx \showarticletitle \undefined \def \showarticletitle #1{#1}   \fi
	\ifx \showURL      \undefined \def \showURL       {\relax}        \fi
	\providecommand\bibfield[2]{#2}
	\providecommand\bibinfo[2]{#2}
	\providecommand\natexlab[1]{#1}
	\providecommand\showeprint[2][]{arXiv:#2}
	
	\bibitem[\protect\citeauthoryear{Akrami, Chaudhury, Mehlhorn, Shahkarami, and
		Vermande}{Akrami et~al\mbox{.}}{2021}]%
	{akrami2021nash}
	\bibfield{author}{\bibinfo{person}{Hannaneh Akrami},
		\bibinfo{person}{Bhaskar~Ray Chaudhury}, \bibinfo{person}{Kurt Mehlhorn},
		\bibinfo{person}{Golnoosh Shahkarami}, {and} \bibinfo{person}{Quentin
			Vermande}.} \bibinfo{year}{2021}\natexlab{}.
	\newblock \bibinfo{title}{Maximizing {N}ash Social Welfare in 2-Value
		Instances}.
	\newblock \bibinfo{howpublished}{arXiv:2107.08965}.
	\newblock
	
	
	\bibitem[\protect\citeauthoryear{Amanatidis, Birmpas, Filos-Ratsikas,
		Hollender, and Voudouris}{Amanatidis et~al\mbox{.}}{2021}]%
	{amanatidis2021maximum}
	\bibfield{author}{\bibinfo{person}{Georgios Amanatidis},
		\bibinfo{person}{Georgios Birmpas}, \bibinfo{person}{Aris Filos-Ratsikas},
		\bibinfo{person}{Alexandros Hollender}, {and} \bibinfo{person}{Alexandros~A
			Voudouris}.} \bibinfo{year}{2021}\natexlab{}.
	\newblock \showarticletitle{Maximum {N}ash welfare and other stories about
		{EFX}}.
	\newblock \bibinfo{journal}{\emph{Theoretical Computer Science}}
	\bibinfo{volume}{863} (\bibinfo{year}{2021}), \bibinfo{pages}{69--85}.
	\newblock
	
	
	\bibitem[\protect\citeauthoryear{Amanatidis, Markakis, Nikzad, and
		Saberi}{Amanatidis et~al\mbox{.}}{2017}]%
	{amanatidis2017approximation}
	\bibfield{author}{\bibinfo{person}{Georgios Amanatidis},
		\bibinfo{person}{Evangelos Markakis}, \bibinfo{person}{Afshin Nikzad}, {and}
		\bibinfo{person}{Amin Saberi}.} \bibinfo{year}{2017}\natexlab{}.
	\newblock \showarticletitle{Approximation algorithms for computing maximin
		share allocations}.
	\newblock \bibinfo{journal}{\emph{ACM Transactions on Algorithms (TALG)}}
	\bibinfo{volume}{13}, \bibinfo{number}{4} (\bibinfo{year}{2017}),
	\bibinfo{pages}{1--28}.
	\newblock
	
	
	\bibitem[\protect\citeauthoryear{Aziz, Bir{\'o}, Lang, Lesca, and Monnot}{Aziz
		et~al\mbox{.}}{2019}]%
	{aziz2019efficient}
	\bibfield{author}{\bibinfo{person}{Haris Aziz}, \bibinfo{person}{P{\'e}ter
			Bir{\'o}}, \bibinfo{person}{J{\'e}r{\^o}me Lang}, \bibinfo{person}{Julien
			Lesca}, {and} \bibinfo{person}{J{\'e}r{\^o}me Monnot}.}
	\bibinfo{year}{2019}\natexlab{}.
	\newblock \showarticletitle{Efficient reallocation under additive and
		responsive preferences}.
	\newblock \bibinfo{journal}{\emph{Theoretical Computer Science}}
	\bibinfo{volume}{790} (\bibinfo{year}{2019}), \bibinfo{pages}{1--15}.
	\newblock
	
	
	\bibitem[\protect\citeauthoryear{Aziz and Brown}{Aziz and Brown}{2020}]%
	{aziz2020random}
	\bibfield{author}{\bibinfo{person}{Haris Aziz} {and} \bibinfo{person}{Ethan
			Brown}.} \bibinfo{year}{2020}\natexlab{}.
	\newblock \bibinfo{title}{Random Assignment Under Bi-Valued Utilities:
		Analyzing Hylland-Zeckhauser, Nash-Bargaining, and other Rules}.
	\newblock \bibinfo{howpublished}{arXiv:2006.15747}.
	\newblock
	
	
	\bibitem[\protect\citeauthoryear{Aziz, Caragiannis, Igarashi, and Walsh}{Aziz
		et~al\mbox{.}}{2022}]%
	{aziz2022fair}
	\bibfield{author}{\bibinfo{person}{Haris Aziz}, \bibinfo{person}{Ioannis
			Caragiannis}, \bibinfo{person}{Ayumi Igarashi}, {and} \bibinfo{person}{Toby
			Walsh}.} \bibinfo{year}{2022}\natexlab{}.
	\newblock \showarticletitle{Fair allocation of indivisible goods and chores}.
	\newblock \bibinfo{journal}{\emph{Autonomous Agents and Multi-Agent Systems}}
	\bibinfo{volume}{36}, \bibinfo{number}{3} (\bibinfo{year}{2022}),
	\bibinfo{pages}{1--21}.
	\newblock
	
	
	\bibitem[\protect\citeauthoryear{Aziz, Rauchecker, Schryen, and Walsh}{Aziz
		et~al\mbox{.}}{2017}]%
	{aziz2017algorithms}
	\bibfield{author}{\bibinfo{person}{Haris Aziz}, \bibinfo{person}{Gerhard
			Rauchecker}, \bibinfo{person}{Guido Schryen}, {and} \bibinfo{person}{Toby
			Walsh}.} \bibinfo{year}{2017}\natexlab{}.
	\newblock \showarticletitle{Algorithms for max-min share fair allocation of
		indivisible chores}. In \bibinfo{booktitle}{\emph{Proceedings of the 31st
			AAAI Conference on Artificial Intelligence (AAAI)}}.
	\bibinfo{pages}{335--341}.
	\newblock
	
	
	\bibitem[\protect\citeauthoryear{Barman and Krishnamurthy}{Barman and
		Krishnamurthy}{2019}]%
	{barman2019proximity}
	\bibfield{author}{\bibinfo{person}{Siddharth Barman} {and}
		\bibinfo{person}{Sanath~Kumar Krishnamurthy}.}
	\bibinfo{year}{2019}\natexlab{}.
	\newblock \showarticletitle{On the proximity of markets with integral
		equilibria}. In \bibinfo{booktitle}{\emph{Proceedings of the 33rd AAAI
			Conference on Artificial Intelligence (AAAI)}}. \bibinfo{pages}{1748--1755}.
	\newblock
	
	
	\bibitem[\protect\citeauthoryear{Barman, Krishnamurthy, and Vaish}{Barman
		et~al\mbox{.}}{2018a}]%
	{barman2018finding}
	\bibfield{author}{\bibinfo{person}{Siddharth Barman},
		\bibinfo{person}{Sanath~Kumar Krishnamurthy}, {and} \bibinfo{person}{Rohit
			Vaish}.} \bibinfo{year}{2018}\natexlab{a}.
	\newblock \showarticletitle{Finding fair and efficient allocations}. In
	\bibinfo{booktitle}{\emph{Proceedings of the 2018 ACM Conference on Economics
			and Computation (EC)}}. \bibinfo{pages}{557--574}.
	\newblock
	
	
	\bibitem[\protect\citeauthoryear{Barman, Krishnamurthy, and Vaish}{Barman
		et~al\mbox{.}}{2018b}]%
	{barman2018greedy}
	\bibfield{author}{\bibinfo{person}{Siddharth Barman},
		\bibinfo{person}{Sanath~Kumar Krishnamurthy}, {and} \bibinfo{person}{Rohit
			Vaish}.} \bibinfo{year}{2018}\natexlab{b}.
	\newblock \showarticletitle{Greedy Algorithms for Maximizing {Nash} Social
		Welfare}. In \bibinfo{booktitle}{\emph{Proceedings of the 17th International
			Conference on Autonomous Agents and Multiagent Systems (AAMAS)}}.
	\bibinfo{pages}{7--13}.
	\newblock
	
	
	\bibitem[\protect\citeauthoryear{Bogomolnaia, Moulin, Sandomirskiy, and
		Yanovskaya}{Bogomolnaia et~al\mbox{.}}{2017}]%
	{bogomolnaia2017competitive}
	\bibfield{author}{\bibinfo{person}{Anna Bogomolnaia},
		\bibinfo{person}{Herv{\'e} Moulin}, \bibinfo{person}{Fedor Sandomirskiy},
		{and} \bibinfo{person}{Elena Yanovskaya}.} \bibinfo{year}{2017}\natexlab{}.
	\newblock \showarticletitle{Competitive division of a mixed manna}.
	\newblock \bibinfo{journal}{\emph{Econometrica}} \bibinfo{volume}{85},
	\bibinfo{number}{6} (\bibinfo{year}{2017}), \bibinfo{pages}{1847--1871}.
	\newblock
	
	
	\bibitem[\protect\citeauthoryear{Boodaghians, Chaudhury, and Mehta}{Boodaghians
		et~al\mbox{.}}{2022}]%
	{boodaghians2022polynomial}
	\bibfield{author}{\bibinfo{person}{Shant Boodaghians},
		\bibinfo{person}{Bhaskar~Ray Chaudhury}, {and} \bibinfo{person}{Ruta Mehta}.}
	\bibinfo{year}{2022}\natexlab{}.
	\newblock \showarticletitle{Polynomial Time Algorithms to Find an Approximate
		Competitive Equilibrium for Chores}. In \bibinfo{booktitle}{\emph{Proceedings
			of the 2022 Annual ACM-SIAM Symposium on Discrete Algorithms (SODA)}}. SIAM,
	\bibinfo{pages}{2285--2302}.
	\newblock
	
	
	\bibitem[\protect\citeauthoryear{Bouveret and Lema{\^\i}tre}{Bouveret and
		Lema{\^\i}tre}{2016}]%
	{BL16}
	\bibfield{author}{\bibinfo{person}{Sylvain Bouveret} {and}
		\bibinfo{person}{Michel Lema{\^\i}tre}.} \bibinfo{year}{2016}\natexlab{}.
	\newblock \showarticletitle{Characterizing conflicts in fair division of
		indivisible goods using a scale of criteria}.
	\newblock \bibinfo{journal}{\emph{Autonomous Agents and Multi-Agent Systems}}
	\bibinfo{volume}{30}, \bibinfo{number}{2} (\bibinfo{year}{2016}),
	\bibinfo{pages}{259--290}.
	\newblock
	
	
	\bibitem[\protect\citeauthoryear{Br{\^a}nzei and Sandomirskiy}{Br{\^a}nzei and
		Sandomirskiy}{2019}]%
	{branzei2019algorithms}
	\bibfield{author}{\bibinfo{person}{Simina Br{\^a}nzei} {and}
		\bibinfo{person}{Fedor Sandomirskiy}.} \bibinfo{year}{2019}\natexlab{}.
	\newblock \bibinfo{title}{Algorithms for competitive division of chores}.
	\newblock \bibinfo{howpublished}{arXiv:1907.01766}.
	\newblock
	
	
	\bibitem[\protect\citeauthoryear{Budish}{Budish}{2011}]%
	{Bud11}
	\bibfield{author}{\bibinfo{person}{Eric Budish}.}
	\bibinfo{year}{2011}\natexlab{}.
	\newblock \showarticletitle{The combinatorial assignment problem: Approximate
		competitive equilibrium from equal incomes}.
	\newblock \bibinfo{journal}{\emph{Journal of Political Economy}}
	\bibinfo{volume}{119}, \bibinfo{number}{6} (\bibinfo{year}{2011}),
	\bibinfo{pages}{1061--1103}.
	\newblock
	
	
	\bibitem[\protect\citeauthoryear{Caragiannis, Kurokawa, Moulin, Procaccia,
		Shah, and Wang}{Caragiannis et~al\mbox{.}}{2019}]%
	{CKMP+19}
	\bibfield{author}{\bibinfo{person}{Ioannis Caragiannis}, \bibinfo{person}{David
			Kurokawa}, \bibinfo{person}{Herv{\'e} Moulin}, \bibinfo{person}{Ariel~D
			Procaccia}, \bibinfo{person}{Nisarg Shah}, {and} \bibinfo{person}{Junxing
			Wang}.} \bibinfo{year}{2019}\natexlab{}.
	\newblock \showarticletitle{The unreasonable fairness of maximum Nash welfare}.
	\newblock \bibinfo{journal}{\emph{ACM Transactions on Economics and Computation
			(TEAC)}} \bibinfo{volume}{7}, \bibinfo{number}{3} (\bibinfo{year}{2019}),
	\bibinfo{pages}{1--32}.
	\newblock
	
	
	\bibitem[\protect\citeauthoryear{Chaudhury, Garg, and Mehlhorn}{Chaudhury
		et~al\mbox{.}}{2020}]%
	{chaudhury2020efx}
	\bibfield{author}{\bibinfo{person}{Bhaskar~Ray Chaudhury},
		\bibinfo{person}{Jugal Garg}, {and} \bibinfo{person}{Kurt Mehlhorn}.}
	\bibinfo{year}{2020}\natexlab{}.
	\newblock \showarticletitle{EFX exists for three agents}. In
	\bibinfo{booktitle}{\emph{Proceedings of the 2020 ACM Conference on Economics
			and Computation (EC)}}. \bibinfo{pages}{1--19}.
	\newblock
	
	
	\bibitem[\protect\citeauthoryear{Devanur, Papadimitriou, Saberi, and
		Vazirani}{Devanur et~al\mbox{.}}{2008}]%
	{devanur2008market}
	\bibfield{author}{\bibinfo{person}{Nikhil~R Devanur},
		\bibinfo{person}{Christos~H Papadimitriou}, \bibinfo{person}{Amin Saberi},
		{and} \bibinfo{person}{Vijay~V Vazirani}.} \bibinfo{year}{2008}\natexlab{}.
	\newblock \showarticletitle{Market equilibrium via a primal--dual algorithm for
		a convex program}.
	\newblock \bibinfo{journal}{\emph{Journal of the ACM (JACM)}}
	\bibinfo{volume}{55}, \bibinfo{number}{5} (\bibinfo{year}{2008}),
	\bibinfo{pages}{1--18}.
	\newblock
	
	
	\bibitem[\protect\citeauthoryear{Foley}{Foley}{1967}]%
	{Fol67}
	\bibfield{author}{\bibinfo{person}{Duncan~Karl Foley}.}
	\bibinfo{year}{1967}\natexlab{}.
	\newblock \bibinfo{booktitle}{\emph{Resource allocation and the public
			sector}}.
	\newblock \bibinfo{publisher}{Yale University}.
	\newblock
	
	
	\bibitem[\protect\citeauthoryear{Gamow and Stern}{Gamow and Stern}{1958}]%
	{GS58}
	\bibfield{author}{\bibinfo{person}{George Gamow} {and} \bibinfo{person}{Marvin
			Stern}.} \bibinfo{year}{1958}\natexlab{}.
	\newblock \bibinfo{booktitle}{\emph{Puzzle-Math}}.
	\newblock \bibinfo{publisher}{Viking}.
	\newblock
	
	
	\bibitem[\protect\citeauthoryear{Garey and Johnson}{Garey and Johnson}{1990}]%
	{GJ79}
	\bibfield{author}{\bibinfo{person}{Michael~R. Garey} {and}
		\bibinfo{person}{David~S. Johnson}.} \bibinfo{year}{1990}\natexlab{}.
	\newblock \bibinfo{booktitle}{\emph{Computers and Intractability: A Guide to
			the Theory of {NP}-Completeness}}.
	\newblock \bibinfo{publisher}{W. H. Freeman \& Co.}
	\newblock
	
	
	\bibitem[\protect\citeauthoryear{Garg, McGlaughlin, and Taki}{Garg
		et~al\mbox{.}}{2019}]%
	{garg2019approximating}
	\bibfield{author}{\bibinfo{person}{Jugal Garg}, \bibinfo{person}{Peter
			McGlaughlin}, {and} \bibinfo{person}{Setareh Taki}.}
	\bibinfo{year}{2019}\natexlab{}.
	\newblock \showarticletitle{Approximating Maximin Share Allocations}. In
	\bibinfo{booktitle}{\emph{Proceedings of the 2nd Symposium on Simplicity in
			Algorithms ({SOSA})}}, Vol.~\bibinfo{volume}{69}.
	\bibinfo{pages}{20:1--20:11}.
	\newblock
	
	
	\bibitem[\protect\citeauthoryear{Garg and Murhekar}{Garg and Murhekar}{2021a}]%
	{garg2021computing}
	\bibfield{author}{\bibinfo{person}{Jugal Garg} {and} \bibinfo{person}{Aniket
			Murhekar}.} \bibinfo{year}{2021}\natexlab{a}.
	\newblock \showarticletitle{Computing Fair and Efficient Allocations with Few
		Utility Values}. In \bibinfo{booktitle}{\emph{International Symposium on
			Algorithmic Game Theory}}. \bibinfo{pages}{345--359}.
	\newblock
	
	
	\bibitem[\protect\citeauthoryear{Garg and Murhekar}{Garg and Murhekar}{2021b}]%
	{garg2021fair}
	\bibfield{author}{\bibinfo{person}{Jugal Garg} {and} \bibinfo{person}{Aniket
			Murhekar}.} \bibinfo{year}{2021}\natexlab{b}.
	\newblock \showarticletitle{On Fair and Efficient Allocations of Indivisible
		Goods}. In \bibinfo{booktitle}{\emph{Proceedings of the 35th AAAI Conference
			on Artificial Intelligence (AAAI)}}. \bibinfo{pages}{5595--5602}.
	\newblock
	
	
	\bibitem[\protect\citeauthoryear{Garg, Murhekar, and Qin}{Garg
		et~al\mbox{.}}{2022}]%
	{bivaluedAAAI}
	\bibfield{author}{\bibinfo{person}{Jugal Garg}, \bibinfo{person}{Aniket
			Murhekar}, {and} \bibinfo{person}{John Qin}.}
	\bibinfo{year}{2022}\natexlab{}.
	\newblock \showarticletitle{Fair and Efficient Allocations of Chores under
		Bivalued Preferences}. In \bibinfo{booktitle}{\emph{Proceedings of the 36th
			AAAI Conference on Artificial Intelligence (AAAI)}}.
	\newblock
	
	
	\bibitem[\protect\citeauthoryear{Garg and Taki}{Garg and Taki}{2021}]%
	{garg2021improved}
	\bibfield{author}{\bibinfo{person}{Jugal Garg} {and} \bibinfo{person}{Setareh
			Taki}.} \bibinfo{year}{2021}\natexlab{}.
	\newblock \showarticletitle{An improved approximation algorithm for maximin
		shares}.
	\newblock \bibinfo{journal}{\emph{Artificial Intelligence}}
	\bibinfo{volume}{300} (\bibinfo{year}{2021}), \bibinfo{pages}{103547}.
	\newblock
	
	
	\bibitem[\protect\citeauthoryear{Ghodsi, HajiAghayi, Seddighin, Seddighin, and
		Yami}{Ghodsi et~al\mbox{.}}{2018}]%
	{ghodsi2018fair}
	\bibfield{author}{\bibinfo{person}{Mohammad Ghodsi},
		\bibinfo{person}{MohammadTaghi HajiAghayi}, \bibinfo{person}{Masoud
			Seddighin}, \bibinfo{person}{Saeed Seddighin}, {and} \bibinfo{person}{Hadi
			Yami}.} \bibinfo{year}{2018}\natexlab{}.
	\newblock \showarticletitle{Fair allocation of indivisible goods: Improvements
		and generalizations}. In \bibinfo{booktitle}{\emph{Proceedings of the 19th
			ACM Conference on Economics and Computation (EC)}}.
	\bibinfo{pages}{539--556}.
	\newblock
	
	
	\bibitem[\protect\citeauthoryear{Halpern, Procaccia, Psomas, and Shah}{Halpern
		et~al\mbox{.}}{2020}]%
	{halpern2020fair}
	\bibfield{author}{\bibinfo{person}{Daniel Halpern}, \bibinfo{person}{Ariel~D
			Procaccia}, \bibinfo{person}{Alexandros Psomas}, {and}
		\bibinfo{person}{Nisarg Shah}.} \bibinfo{year}{2020}\natexlab{}.
	\newblock \showarticletitle{Fair division with binary valuations: One rule to
		rule them all}. In \bibinfo{booktitle}{\emph{Proceedings of the 16th
			International Conference on Web and Internet Economics (WINE)}}. Springer,
	\bibinfo{pages}{370--383}.
	\newblock
	
	
	\bibitem[\protect\citeauthoryear{Hosseini, Sikdar, Vaish, and Xia}{Hosseini
		et~al\mbox{.}}{2021}]%
	{HSVX21}
	\bibfield{author}{\bibinfo{person}{Hadi Hosseini}, \bibinfo{person}{Sujoy
			Sikdar}, \bibinfo{person}{Rohit Vaish}, {and} \bibinfo{person}{Lirong Xia}.}
	\bibinfo{year}{2021}\natexlab{}.
	\newblock \showarticletitle{Fair and Efficient Allocations under Lexicographic
		Preferences}. In \bibinfo{booktitle}{\emph{Proceedings of the 35th AAAI
			Conference on Artificial Intelligence (AAAI)}}. \bibinfo{pages}{5472--5480}.
	\newblock
	
	
	\bibitem[\protect\citeauthoryear{Huang and Lu}{Huang and Lu}{2021}]%
	{huang2021algorithmic}
	\bibfield{author}{\bibinfo{person}{Xin Huang} {and} \bibinfo{person}{Pinyan
			Lu}.} \bibinfo{year}{2021}\natexlab{}.
	\newblock \showarticletitle{An algorithmic framework for approximating maximin
		share allocation of chores}. In \bibinfo{booktitle}{\emph{Proceedings of the
			22nd ACM Conference on Economics and Computation (EC)}}.
	\bibinfo{pages}{630--631}.
	\newblock
	
	
	\bibitem[\protect\citeauthoryear{Kurokawa, Procaccia, and Wang}{Kurokawa
		et~al\mbox{.}}{2016}]%
	{kurokawa2016can}
	\bibfield{author}{\bibinfo{person}{David Kurokawa}, \bibinfo{person}{Ariel~D
			Procaccia}, {and} \bibinfo{person}{Junxing Wang}.}
	\bibinfo{year}{2016}\natexlab{}.
	\newblock \showarticletitle{When can the maximin share guarantee be
		guaranteed?}. In \bibinfo{booktitle}{\emph{Proceedings of the 30th AAAI
			Conference on Artificial Intelligence (AAAI)}}. \bibinfo{pages}{523--529}.
	\newblock
	
	
	\bibitem[\protect\citeauthoryear{Kurokawa, Procaccia, and Wang}{Kurokawa
		et~al\mbox{.}}{2018}]%
	{kurokawa2018fair}
	\bibfield{author}{\bibinfo{person}{David Kurokawa}, \bibinfo{person}{Ariel~D
			Procaccia}, {and} \bibinfo{person}{Junxing Wang}.}
	\bibinfo{year}{2018}\natexlab{}.
	\newblock \showarticletitle{Fair enough: Guaranteeing approximate maximin
		shares}.
	\newblock \bibinfo{journal}{\emph{Journal of the ACM (JACM)}}
	\bibinfo{volume}{65}, \bibinfo{number}{2} (\bibinfo{year}{2018}),
	\bibinfo{pages}{1--27}.
	\newblock
	
	
	\bibitem[\protect\citeauthoryear{Lipton, Markakis, Mossel, and Saberi}{Lipton
		et~al\mbox{.}}{2004}]%
	{LMMS04}
	\bibfield{author}{\bibinfo{person}{Richard~J Lipton},
		\bibinfo{person}{Evangelos Markakis}, \bibinfo{person}{Elchanan Mossel},
		{and} \bibinfo{person}{Amin Saberi}.} \bibinfo{year}{2004}\natexlab{}.
	\newblock \showarticletitle{On approximately fair allocations of indivisible
		goods}. In \bibinfo{booktitle}{\emph{Proceedings of the 5th ACM Conference on
			Economics and Computation (EC)}}. \bibinfo{pages}{125--131}.
	\newblock
	
	
	\bibitem[\protect\citeauthoryear{Orlin}{Orlin}{2010}]%
	{orlin2010improved}
	\bibfield{author}{\bibinfo{person}{James~B Orlin}.}
	\bibinfo{year}{2010}\natexlab{}.
	\newblock \showarticletitle{Improved algorithms for computing {F}isher's market
		clearing prices}. In \bibinfo{booktitle}{\emph{Proceedings of the 42nd Annual
			ACM Symposium on Theory of Computing (STOC)}}. \bibinfo{pages}{291--300}.
	\newblock
	
	
	\bibitem[\protect\citeauthoryear{Plaut and Roughgarden}{Plaut and
		Roughgarden}{2020}]%
	{plaut2020almost}
	\bibfield{author}{\bibinfo{person}{Benjamin Plaut} {and} \bibinfo{person}{Tim
			Roughgarden}.} \bibinfo{year}{2020}\natexlab{}.
	\newblock \showarticletitle{Almost envy-freeness with general valuations}.
	\newblock \bibinfo{journal}{\emph{SIAM Journal on Discrete Mathematics}}
	\bibinfo{volume}{34}, \bibinfo{number}{2} (\bibinfo{year}{2020}),
	\bibinfo{pages}{1039--1068}.
	\newblock
	
	
	\bibitem[\protect\citeauthoryear{Shah}{Shah}{2017}]%
	{Shah17}
	\bibfield{author}{\bibinfo{person}{Nisarg Shah}.}
	\bibinfo{year}{2017}\natexlab{}.
	\newblock \showarticletitle{Spliddit: two years of making the world fairer}.
	\newblock \bibinfo{journal}{\emph{XRDS: Crossroads, The ACM Magazine for
			Students}} \bibinfo{volume}{24}, \bibinfo{number}{1} (\bibinfo{year}{2017}),
	\bibinfo{pages}{24--28}.
	\newblock
	
	
	\bibitem[\protect\citeauthoryear{Varian}{Varian}{1974}]%
	{Var74}
	\bibfield{author}{\bibinfo{person}{Hal~R. Varian}.}
	\bibinfo{year}{1974}\natexlab{}.
	\newblock \showarticletitle{Equity, envy and efficiency}.
	\newblock \bibinfo{journal}{\emph{Journal of Economic Theory}}
	\bibinfo{volume}{9} (\bibinfo{year}{1974}), \bibinfo{pages}{63--91}.
	\newblock
	
	
\end{thebibliography}

\newpage
\appendix
\section*{Appendix}
\appendix

\section{Proof of Lemma~\ref{lem:properties-phase2b}}

In this section, we prove the two useful properties of Phase 2b claimed in \Cref{lem:properties-phase2b}. These properties hold for general additive utilities, as shown in \citet{barman2018finding}.
In Phase 2b, we use a relaxed condition in \cref{line:phase-2b-condition} compared to the Phase 2 described in \citet{barman2018finding}. The difference is that for the shortest MPB alternating path $\smash{\ls \stackrel{c_1}{\gets} i_1 \stackrel{c_2}{\gets} \cdots \stackrel{c_\ell}{\gets} i_\ell}$, instead of making a transfer when $\price(\xx_{i_{\ell}}) - \price(c_{\ell}) > \price(\xx_{\ls})$ (referred to as a ``path violator''), we do a transfer when $\priceone(\xx_{i_\ell}) > \price(\xx_{\ls})$ (referred to as a ``violator''). Note that $\price(\xx_{i_{\ell}}) - \price(c_{\ell})  \le \priceone(\xx_{i_\ell})$. Therefore, if our  Phase2b makes a transfer then the original Phase 2 in \cite{barman2018finding} also does that.

One observation is that the minimum spending value never decreases over the run of a Phase 2b. This is \Cref{lem:properties-phase2b} (\ref{enum:phase2b-ls-spending}).

\begin{lemma}
	\label{lem:ls-spending-increase}
	During a run of Phase 2b, the minimum spending, i.e.\ $\min_{i \in \agents} \price(\xx_i)$, never decreases. 
\end{lemma}
\begin{proof}
	It is sufficient to show the minimum spending does not decrease after each single transfer of Phase 2b. Suppose we transfer a chore $c$ from agent $i$ to agent $j$. Let $\ls$ be the least spender before the transfer. Let $\xx$ and $\xx'$ denote the allocations before and after the transfer respectively. Right before the transfer, $i$ must have been a violator to $\ls$, then
	$
	\price(\xx_{\ls}) < \priceone(\xx_i) \le \price(\xx_i) - \price(c) = \price(\xx'_i).
	$
	That is, the spending of $i$ after the transfer is strictly larger than the minimum spending before the transfer. Furthermore, spending of $j$ has increased by $\price(c)$, and other agents have the same bundle (and spending) as they had before the transfer. In conclusion, $\min_{i' \in \agents }\price(\xx_{i'}) \le \min_{i' \in \agents } \price(\xx'_{\ls})$.
\end{proof}

Next, we show that Phase 2b must terminate after at most $\poly(n, m, \max_{i \in \agents} \abs{\mathcal{U}_i})$ steps, where $U_i$ is the set of all different utilities agent $i$ has for all subsets of items. 
The proof follows from the next two lemmas.

\begin{lemma}[Lemma 13, \citet{barman2018finding}]
	\label{lem:phase2-ls-change}
	After $\poly(n, m)$ steps in Phase 2b, either the identity of the least spender changes or a Phase 3 happens.
\end{lemma}
\begin{proof}

	At time $t$, let $LS$ be the set of agents with the minimum spending, and for all agents $i \in \agents$ define
	\[
	\level(i, t) \coloneq
	\begin{cases}
		\ell , & \text{if $\exists \ls \in LS \colon \ls \reachable i$, and $\ell$ is the length of the shortest of such paths,} \\
		n, & \text{if $\not\exists \ls \in LS \colon \ls \reachable i$.}
	\end{cases}
	\]
	
	Furthermore, let $G(i, t)$ be the set of chores $c \in \xx^t_i$ such that there exists an MPB alternating path $\smash{\ls {\gets} \cdots \gets i' \stackrel{c}{\gets} i}$ where the last edge uses chore $c$.
	For agents $i$ where $\level(i, t) = n$, $G(i, t) = \emptyset$.
	
	Now, define the potential function $\phi(t)$ as follows,
	\[
	\phi(t) = \sum_{i \in \agents} m \cdot (n - \level(i, t)) + \len{G_{i, t}}.
	\]
	Note that $\phi$ is always integral and positive, $\sum_{i \in \agents} \len{G_{i, t}} \le m$, and $\sum_{i \in \agents} m \cdot (n - \level(i, t)) \le mn^2$. Then, to show the lemma holds, it suffices to prove the potential function strictly decreases after each transfer. Therefore, Phase 2b terminates after $O(mn^2)$ steps. 
	
	Suppose we transfer the chore $c$ from  agent $i_\ell$ to agent $i_{\ell - 1}$.

	Agents at level $\ell - 2$ do not consider $c$ their MPB, because otherwise the shortest path to $i_\ell$ would be of length $\ell -1$. Therefore, after $i_{\ell - 1}$ receives $c$, $\level(i_{\ell - 1}, t + 1) = \level(i_{\ell - 1}, t)$ and $G(i_{\ell - 1}, t + 1) = G(i_{\ell - 1}, t)$, and $\phi$ does not change for the terms related to $i_{\ell - 1}$.
	
	Similarly, other agents $i' \in \agents \setminus \set{i_{\ell}, i_{\ell - 1}}$ cannot move to lower levels after this transfer.

	For $i_\ell$, either there exists another chore in $G_{i, t}$ which keeps her in level $\ell$, or she moves to a strictly higher level (or possibly $\level(i_{\ell}, t + 1) = n$). In either case, we can show $\phi$ strictly decreases.

	If $\level(t, i_\ell) < \level(t + 1, i_\ell)$, then any change in $\sum_{i \in \agents} \len{G_{i, t}}$ will be cancelled out by the decrease in $m \cdot (n - \level(i_\ell, t))$ due to the lexicographical weighting. 
	
	If $\level(t, i_\ell) = \level(t + 1, i_\ell)$, then for other agents $i' \in \agents \setminus \set{i_{\ell}, i_{\ell - 1}}$, $G_{i', t}$ does not change. However, $\abs{G_{i_\ell, t + 1}} = \abs{G_{i_{\ell}, t}} - 1$. Thus, $f$ decreases by at least one after each transfer in Phase 2b.
\end{proof}

\begin{lemma}
	\label{lem:phase2-same-ls}
	During a continuous run of Phase 2b, if agent $i$ ceases being the least spender after time $t$, and becomes the least spender again at some time $t' > t$, then her utility have must have decreased, i.e.\ ${v_i(\xx^{t'}_i)} < v_i(\xx^{t}_i)$.
\end{lemma}
\begin{proof}
	When agent $i$ ceases being the least spender, she must have beenshe must have received a chore $c$ at time $t$. That is $\price(\xx^t_i) = \min_{i' \in \agents} \price(\xx^t_{i'})$, and $\xx^{t + 1}_i = \xx^t_i \cup \set{c}$. First, suppose $\xx^{t + 1}_i \subsetneq \xx^{t'}_i$, i.e.\ $i$ did not give away any chores from $t + 1$ to $t'$, then $v_i(\xx^{t'}_i) \le v_i(\xx^{t + 1}_i) < v_i(\xx^{t}_i)$.

	Now, assume $i$ has given away at least one chore from $t + 1$ to $t'$. Let $t_\ell$ be the last time she gave away an item. Suppose that item is $c'$. At $t_\ell$, $i$ must have been a violator to the least spender, say agent $\ls$, then
	\[
	\price(\xx^{t_\ell}_{\ls}) < \priceone(\xx^{t_\ell}_{i}) \le \price(\xx^{t_\ell}_i) - \price(c') = \price(\xx^{t_\ell + 1}_i).
	\]
	Furthermore, as $i$ was the least spender at $t$ and the minimum spending has not decreased by \Cref{lem:ls-spending-increase}, $\price(\xx^t_i) \le \price(\xx^{t_\ell}_{\ls})$. Putting these together, we conclude that $\price(\xx^t_i) < \price(\xx^{t_\ell + 1}_i)$. Since this was the last time $i$ gave away a chore, her spending could not go any lower. Thus, $\price(\xx^{t}_i) < \price(\xx^{t'}_i)$. There were no price changes between $t$ and $t'$, then $\MPB_i$ has remained the same, and $\abs{v_i(\xx^t_i)} = \MPB_i \cdot \price(\xx^{t}_i) < \MPB_i \cdot \price(\xx^{t'}_i)  = \abs{v_i(\xx^{t'}_i)}$ which completes the proof.
\end{proof}

With the two lemmas above, we can prove an upper bound on the running time of Phase 2b.
\begin{lemma}
	\label{lem:phase2-pseudopoly}
	Phase 2b of \Cref{alg:full-algo} should terminate after at most  $\poly(n, m, \max_{i \in \agents} \len{\mathcal{U}_i})$ time, where $\mathcal{U}_i = \set{v_i(S) \mid \forall S \subseteq \items}$ is the set of all different utilities agent $i$ can obtain.
\end{lemma}
\begin{proof}
	By \Cref{lem:phase2-same-ls}, the number of times an agent ceases being among the least spenders is bounded by the number of different utilities she can have, i.e.\,  $\max_{i \in \agents} \len{\mathcal{U}_i}$. Moreover, by \Cref{lem:phase2-ls-change}, after $\poly(n, m)$ time the identity of the least spender must change or we turn to a Phase 3. Therefore, a continuous Phase 2 can run for at most $\poly(n, m, \max_{i \in \agents} \len{\mathcal{U}_i})$
\end{proof}

As a corollary of \Cref{lem:phase2-pseudopoly} applied to the bivalued chores case, and due to the fact that $\len{\mathcal{U}_i} \le m^2$ (fix the number of $-1$'s and $-p$'s in the bundle), Phase 2b terminates in $\poly(n, m)$ time.
Therefore, we have proved part (\ref{enum:phase2b-poly}) of \Cref{lem:properties-phase2b} as \Cref{lem:phase2-pseudopoly} and part (\ref{enum:phase2b-ls-spending}) of \Cref{lem:properties-phase2b} as \Cref{lem:ls-spending-increase}. 

\newpage
\section{PO is equivalent to $\text{f}$PO for bivalued utilities}
\label{sec:fpo}

\newcommand{\zz}{\mathbf{z}}

In this section, we will prove \Cref{thm:fpo}, which states that for bivalued utilities, an allocation $\xx$ is Pareto optimal (PO) if and only if it is fractionally Pareto optimal (fPO). We will give the proof for the case of chore division. The proof for goods division is very similar, and can be obtained by reversing the direction of the arrows in the pictures and by swapping the terms ``give away'' and ``receive''. It is worth noting that \Cref{thm:fpo} does \emph{not} hold for personalized bivalued utilities. A counterexample for two agents and two goods is $v_1(a) = 1$, $v_1(b) = 2$, $v_2(a) = 1$, and $v_2(b) = 3$. Then the allocation $\xx_1 = \{b\}$, $\xx_2 = \{a\}$ is Pareto optimal, but it fails fPO since it is dominated by the fractional allocation $\xx_1 = \{a, \frac12b\}$, $\xx_2 = \{\frac12b\}$.

First, note that one direction is trivial: an fPO allocation is also PO.
For the other direction, let $\xx$ be an (integral) allocation that fails fPO. We will show that $\xx$ is not (integrally) Pareto optimal. Fix an arbitrary fractional allocation $\yy$ that Pareto dominates $\xx$. (A \emph{fractional allocation} $\yy = (\yy_{i,c})_{i \in \agents, c \in \items}$ is a collection of numbers $\yy_{i,c} \in [0,1]$ with $\sum_{i\in \agents} \yy_{i,c} = 1$ for all $c\in \items$, where $\yy_{i,c}$ denotes the fraction of item $c$ allocated to agent $i$.)

Fixing the initial allocation $\xx$, any fractional allocation $\zz$ can be described by a ``diff'' vector $(\alpha_{i,j,c})_{i,j\in \agents, c\in \items}$ where $\alpha_{i,j,c} \in [0,1]$ describes the amount of item $c$ that agent $i$ needs to give to $j$ in order to turn $\xx$ into $\zz$. Thus
\[
\alpha_{i,j,c} = \begin{cases}
	0 & \text{if $i = j$}, \\
	0 & \text{if $i \neq j$ and $c \not\in \xx_i$}, \\
	z_{j,c} & \text{if $i \neq j$ and  $c \in \xx_i$}. \\
\end{cases}
\]
From now on let $\zz$ be a fractional allocation such that $v_{i}(\zz) \ge v_{i}(\yy)$ for all $i \in \agents$, and among such allocations, let $\zz$ be the one that minimizes $\sum_{i,j,c} \alpha_{i,j,c}$. Note that such an allocation $\zz$ exists, because the objective function is continuous and the feasibility set ($v_{i}(\zz) \ge v_{i}(\yy)$ for all $i$) is compact. Because $\yy$ Pareto dominates $\xx$, then $\zz$ also Pareto dominates $\xx$.

When $\alpha_{i,j,c} > 0$, we draw the following edge, indicating that $c$ was (partially) transferred, and showing the valuation of both agents for $c$.
\[
\begin{tikzpicture}[]
	\node (i1) {$i$};
	\node (i2) [right=2.5cm of i1] {$j$};
	\draw[-latex] 
		(i1) -- (i2) 
		node [midway, font=\small, inner sep=2pt, fill=white] {$c$}
		node [pos=0.1, below, font=\small] {$v_i(c)$}
		node [pos=0.9, below, font=\small] {$v_j(c)$};
\end{tikzpicture}
\]

Since $\zz$ Pareto dominates $\xx$, the utilitarian social welfare of $\zz$ is higher than that of $\xx$. In other words,
\[
	\sum_{i,j,c} (v_j(c) - v_i(c))\cdot \alpha_{i,j,c} > 0,
\]  
because the left-hand side describes the additional social welfare under $\zz$ compared to $\xx$.
It follows that there are $i_1,i_2,c_1$ with $\alpha_{i_1,i_2,c_1} > 0$ such that $v_{i_1}(c_1) = -p$ and $v_{i_2}(c_1) = -1$, that is, at least one chore is partially transferred from an agent $i_1$ who thinks it's difficult to another agent $i_2$ who thinks it's easy.
Thus,
\[
\begin{tikzpicture}[]
	\node (i1) {$i_1$};
	\node (i2) [right=2.5cm of i1] {$i_2$};
	\draw[-latex] 
		(i1) -- (i2) 
		node [midway, font=\small, inner sep=1pt, fill=white] {$c_1$}
		node [pos=0.1, below, font=\small] {$-p$}
		node [pos=0.9, below, font=\small] {$-1$};
\end{tikzpicture}
\]

Now consider a sequence $(i_1, i_2, \dots, i_t)$ of distinct agents of maximum length such that there are chores $c_2, \dots, c_{t-1}$ with
\[
\begin{tikzpicture}[]
	\node (i1) {$i_j$};
	\node (i2) [right=2.5cm of i1] {$i_{j+1}$};
	\draw[-latex] 
		(i1) -- (i2) 
		node [midway, font=\small, inner sep=1pt, fill=white] {$c_j$}
		node [pos=0.1, below, font=\small] {$-1$}
		node [pos=0.9, below, font=\small] {$-1$};
	\node [right=0.5cm of i2] {for $j = 2, \dots, t-1$.};
\end{tikzpicture}
\]
Now there are two possibilities: either the chain stops at agent $i_t$ and we cannot extend it further, or we can extend it to an agent $i_t$ that already appeared in the sequence (so $i_t = i_\ell$ for some $\ell \in \{1,\dots,t-1\}$).

\begin{itemize}
	\item \emph{Case 1: Suppose that the chain cannot be extended.} 
	Consider agent $i_t$. In moving from $\xx$ to $\zz$, she received some of chore $c_{t-1}$ which makes her worse off. Because $\zz$ is a Pareto improvement, she must give away part of at least some chore, say $c_t$, to some agent $i_{t+1}$.
	\begin{itemize}
		\item Suppose $v_{i_t}(c_t) = -p$, i.e. $c_t$ is difficult for $i_t$.
		Thus, we have the following situation:
		
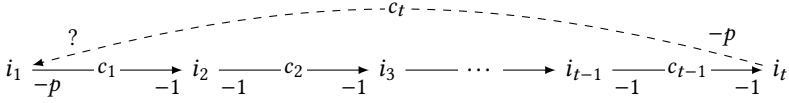
\[
		\begin{tikzpicture}
			\node (i1) {$i_1$};
			\node (i2) [right=2cm of i1] {$i_2$};
			\node (i3) [right=2cm of i2] {$i_3$};
			\node (it1) [right=2cm of i3] {$i_{t-1}$};
			\node (it) [right=2cm of it1] {$i_{t}$};
			\draw[-latex] 
				(i1) -- (i2) 
				node [midway, font=\small, inner sep=1pt, fill=white] {$c_1$}
				node [pos=0.1, below, font=\small] {$-p$}
				node [pos=0.9, below, font=\small] {$-1$};
			\draw[-latex] 
				(i2) -- (i3) 
				node [midway, font=\small, inner sep=1pt, fill=white] {$c_2$}
				node [pos=0.1, below, font=\small] {$-1$}
				node [pos=0.9, below, font=\small] {$-1$};
			\draw[-latex] (i3) -- (it1) node [midway,fill=white] {\dots};

			\draw[-latex] 
				(it1) -- (it) 
				node [midway, font=\small, inner sep=1pt, fill=white] {$c_{t-1}$}
				node [pos=0.1, below, font=\small] {$-1$}
				node [pos=0.9, below, font=\small] {$-1$};
				
			\draw[-latex, dashed] 
				(it) to [bend right=15] 
				node [midway, font=\small, inner sep=1pt, fill=white] {$c_t$}
				node [pos=0.05, above, font=\small] {$-p$}
				node [pos=0.95, above, font=\small] {$?$}
				(i1);
		\end{tikzpicture}
		\]
		(The question mark indicates that we have not determined the value $v_{i_1}(c_t)$, and the dashed line indicates that we may have $\alpha_{i_t, i_1, c_t} = 0$.)
		It follows now that $\xx$ is not (integrally) Pareto optimal. An integral Pareto improvement can be found by implementing the shown cycle ``integrally'': agent $i_j$ gives all of chore $c_j$ to $i_{j+1}$ for $j = 1, \dots, t-1$, and $i_t$ gives all of chore $c_t$ to $i_1$. This change makes $i_t$ strictly better off (receiving an easy chore but giving away a difficult one), leaves $i_2, \dots, i_{t-1}$ indifferent, and either leaves $i_1$ indifferent or makes $i_1$ strictly better off, depending on the value of $v_{i_1}(c_t)$.
		\item Suppose $v_{i_t}(c_t) = -1$, i.e. $c_t$ is easy for $i_t$. 
		Because we cannot extend the chain, it must be that $v_{i_{t+1}}(c_t) = -p$. 
		Thus, we have the following situation:
		\[
		\begin{tikzpicture}
			\node (i1) {$i_{t-1}$};
			\node (i2) [right=2cm of i1] {$i_t$};
			\node (i3) [right=2cm of i2] {$i_{t+1}$};
			\draw[-latex] 
				(i1) -- (i2) 
				node [midway, font=\small, inner sep=1pt, fill=white] {$c_{t-1}$}
				node [pos=0.1, below, font=\small] {$-1$}
				node [pos=0.9, below, font=\small] {$-1$};
			\draw[-latex] 
				(i2) -- (i3) 
				node [midway, font=\small, inner sep=1pt, fill=white] {$c_t$}
				node [pos=0.1, below, font=\small] {$-1$}
				node [pos=0.9, below, font=\small] {$-p$};
		\end{tikzpicture}
		\]
		Let $\beta = \min\{\alpha_{i_{t-1},i_t,c_{t-1}}, \alpha_{i_t,i_{t+1},c_t}\}$. The picture says that $\beta > 0$. Now consider the allocation $\zz'$ that is like $\zz$ except that 
		\begin{alignat*}{5}
			&\alpha'_{i_{t-1},i_t,c_{t-1}} &&= \alpha_{i_{t-1},i_t,c_{t-1}} &&- \beta, \\
			&\alpha'_{i_t,i_{t+1},c_t} &&= \alpha_{i_t,i_{t+1},c_t} &&- \beta, \\
			&\alpha'_{i_{t-1},i_{t+1},c_{t-1}} &&= \alpha_{i_{t-1},i_{t+1},c_{t-1}} &&+ \beta.
		\end{alignat*}
		That is, we reduced the amount transferred along the arcs shown in the picture by $\beta$ and instead transfer a $\beta$ amount of $c_{t-1}$ directly from $i_{t-1}$ to $i_{t+1}$ (i.e., skipping $i_t$).
		Note that agents $i_{t-1}$ and $i_t$ are indifferent between $\zz$ and $\zz'$, and $i_{t+1}$ is either indifferent or is better off in $\zz'$, depending on the value $v_{i_{t+1}}(c_{t-1})$. Thus, for all $i \in \agents$, we have $v_i(\zz') \ge v_i(\zz) \ge v_i(\yy)$. But $\sum_{i,j,c} \alpha'_{i,j,c} < \sum_{i,j,c} \alpha_{i,j,c}$, contradicting our choice of $\zz$.
		\end{itemize}
	\item \emph{Case 2: Suppose that the chain can be extended by repeating an agent.} That is, we can draw another edge to an agent $i_\ell$ with $\ell \in [t-1]$.
		\begin{itemize}
		\item Suppose $\ell = 1$. Thus, we have the following situation:
		
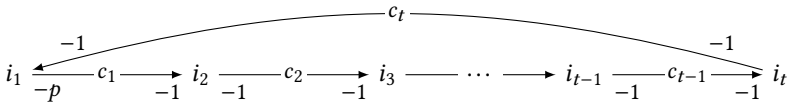
\[
		\begin{tikzpicture}
			\node (i1) {$i_1$};
			\node (i2) [right=2cm of i1] {$i_2$};
			\node (i3) [right=2cm of i2] {$i_3$};
			\node (it1) [right=2cm of i3] {$i_{t-1}$};
			\node (it) [right=2cm of it1] {$i_{t}$};
			\draw[-latex] 
				(i1) -- (i2) 
				node [midway, font=\small, inner sep=1pt, fill=white] {$c_1$}
				node [pos=0.1, below, font=\small] {$-p$}
				node [pos=0.9, below, font=\small] {$-1$};
			\draw[-latex] 
				(i2) -- (i3) 
				node [midway, font=\small, inner sep=1pt, fill=white] {$c_2$}
				node [pos=0.1, below, font=\small] {$-1$}
				node [pos=0.9, below, font=\small] {$-1$};
			\draw[-latex] (i3) -- (it1) node [midway,fill=white] {\dots};

			\draw[-latex] 
				(it1) -- (it) 
				node [midway, font=\small, inner sep=1pt, fill=white] {$c_{t-1}$}
				node [pos=0.1, below, font=\small] {$-1$}
				node [pos=0.9, below, font=\small] {$-1$};
				
			\draw[-latex] 
				(it) to [bend right=15] 
				node [midway, font=\small, inner sep=1pt, fill=white] {$c_t$}
				node [pos=0.05, above, font=\small] {$-1$}
				node [pos=0.95, above, font=\small] {$-1$}
				(i1);
		\end{tikzpicture}
		\]
		It follows that $\xx$ is not (integrally) Pareto optimal. An integral Pareto improvement can be found by implementing the shown cycle ``integrally'': agent $i_j$ gives all of chore $c_j$ to $i_{j+1}$ for $j = 1, \dots, t-1$, and $i_t$ gives all of chore $c_t$ to $i_1$. This change makes $i_1$ strictly better off (receiving an easy chore but giving away a difficult one), and leaves all other agents indifferent.
		\item Suppose $\ell > 1$. For concreteness, we take $\ell = 2$ but the other cases are analogous.
		
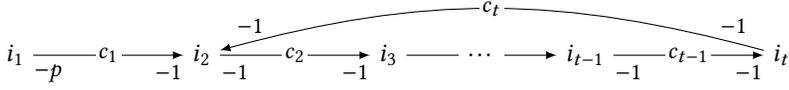
\[
		\begin{tikzpicture}
			\node (i1) {$i_1$};
			\node (i2) [right=2cm of i1] {$i_2$};
			\node (i3) [right=2cm of i2] {$i_3$};
			\node (it1) [right=2cm of i3] {$i_{t-1}$};
			\node (it) [right=2cm of it1] {$i_{t}$};
			\draw[-latex] 
				(i1) -- (i2) 
				node [midway, font=\small, inner sep=1pt, fill=white] {$c_1$}
				node [pos=0.1, below, font=\small] {$-p$}
				node [pos=0.9, below, font=\small] {$-1$};
			\draw[-latex] 
				(i2) -- (i3) 
				node [midway, font=\small, inner sep=1pt, fill=white] {$c_2$}
				node [pos=0.1, below, font=\small] {$-1$}
				node [pos=0.9, below, font=\small] {$-1$};
			\draw[-latex] (i3) -- (it1) node [midway,fill=white] {\dots};
		
			\draw[-latex] 
				(it1) -- (it) 
				node [midway, font=\small, inner sep=1pt, fill=white] {$c_{t-1}$}
				node [pos=0.1, below, font=\small] {$-1$}
				node [pos=0.9, below, font=\small] {$-1$};
			\draw[-latex] 
				(it) to [bend right=15] 
				node [midway, font=\small, inner sep=1pt, fill=white] {$c_t$}
				node [pos=0.05, above, font=\small] {$-1$}
				node [pos=0.95, above, font=\small] {$-1$}
				(i2);
		\end{tikzpicture}
		\]
		Let $\beta = \min\{\alpha_{i_2,i_3,c_2},\dots,\alpha_{i_{t-1},i_t,c_{t-1}}, \alpha_{i_t,i_2,c_t}\}$. The picture says that $\beta > 0$. Now consider the allocation $\zz'$ that is like $\zz$ except that
		\begin{alignat*}{5}
			&\alpha'_{i_j, i_{j+1}, c_j} &&= \alpha_{i_j, i_{j+1}, c_j} &&- \beta \qquad \text{for $j = 2, \dots, t-1$, and} \\
			&\alpha'_{i_t, i_2, c_t} &&= \alpha_{i_t, i_2, c_t} &&- \beta.
		\end{alignat*}
		That is, we reduced the amount transferred along the arcs of the cycle shown in the picture by $\beta$.
		Note that all agents are indifferent between $\zz$ and $\zz'$, because for each edge in the cycle, the receiver values the item received equally to the item given away. Thus for all $i \in \agents$, we have $v_i(\zz') = v_i(\zz) \ge v_i(\yy)$. But $\sum_{i,j,c} \alpha'_{i,j,c} < \sum_{i,j,c} \alpha_{i,j,c}$, contradicting our choice of $\zz$.
		\end{itemize}
\end{itemize}
Each case has led to either a contradiction or else to the desired conclusion that $\xx$ fails Pareto optimality. \qed

\end{document}